\documentclass{amsart}
\usepackage{amssymb,graphicx}
\usepackage{mathrsfs}
\usepackage{color}
\usepackage{enumerate}
\usepackage[colorlinks=true,citecolor=blue,linkcolor=blue]{hyperref}

\newcommand{\Rl}{\mathbb{R}}
\newcommand{\Cplx}{\mathbb{C}}
\newcommand{\Itgr}{\mathbb{Z}}

\newcommand{\Bc}{\mathcal{B}}

\newcommand{\Lc}{\mathcal{L}}

\newcommand{\tr}{\mathrm{Tr}}

\newcommand{\volSd}{\omega_d}

\newcommand{\mbR}{\mathbb{R}}
\newcommand{\mbC}{\mathbb{C}}
\newcommand{\mbZ}{\mathbb{Z}}
\newcommand{\mbN}{\mathbb{N}}

\newcommand{\abs}[1]{\left|#1\right|}
\newcommand{\set}[1]{\{#1\}}
\newcommand{\brs}[1]{\left(#1\right)}

\newcommand{\Scal}[1]{\left\langle #1\right\rangle}               
\newcommand{\norm}[1]{\left \Vert #1 \right\Vert}

\newcommand{\clB}{\mathcal{B}}

\newcommand{\clK}{\mathcal{K}}
\newcommand{\clL}{\mathcal{L}}
\newcommand{\clH}{\mathcal{H}}

\newcommand{\hilb}{\mathcal{H}}

\def\XXint#1#2#3{{\setbox0=\hbox{$#1{#2#3}{\int}$ }
\vcenter{\hbox{$#2#3$ }}\kern-.6\wd0}}

\numberwithin{equation}{section}

\newtheorem{theorem}{Theorem}[section]
\newtheorem{proposition}[theorem]{Proposition}

\newtheorem{definition}[theorem]{Definition}
\newtheorem{lemma}[theorem]{Lemma}

\theoremstyle{remark}
\newtheorem{example}[theorem]{Example}
\newtheorem{remark}[theorem]{Remark}

\title{An application of singular traces to crystals and percolation}
\author{N.\,Azamov, E.\,Hekkelman, E.\,McDonald, F.\,Sukochev, D.\,Zanin}
\date{\today}

\begin{document}
\begin{abstract}
    For a certain class of discrete metric spaces, we provide a formula for the density of states. This formula involves Dixmier traces and is proven using recent advances in operator theory. Various examples are given of metric spaces for which this formula holds, including crystals, quasicrystals and the infinite cluster resulting from super-critical bond percolation on $\mathbb{Z}^d$.
\end{abstract}

\maketitle{}

\tableofcontents

\section{Introduction}
The density of states (DOS) of a Schr\"odinger operator is one of the most useful objects of study in solid-state physics. It is meant to provide an insight in the thermal and electrical conductive properties of a material, and has since been widely applied to study a large variety of physical phenomenon.
Its practical use lies mainly in analysing (extremely) large collections of particles exhibiting a large degree of symmetry, and usually such symmetry is exploited to calculate the DOS in the first place.

On the mathematical side the DOS has also attracted substantial attention, since it is an abstract object that effectively captures some spectral data of an operator. See for example~\cite{AizenmanWarzel2015,BerezinShubin1991, BourgainKlein2013, CarmonaLacroix1990, DoiIwatsuka2001, PasturFigotin1991, Shiryaev1996, Simon1982}. A notorious aspect of the DOS is that the existence of the DOS cannot be guaranteed in general situations. This paper seeks to provide a formula based on Dixmier traces for the DOS on a broad class of discrete metric spaces that includes crystals, quasicrystals and infinite clusters in $\mathbb{Z}^d$ resulting from percolation. The advantage of this formula is that it is guaranteed to be well-defined even if the DOS itself has not been shown to exist. In fact, it can be interpreted as a generalisation of the definition of the DOS.

The DOS associated with an operator $H$ is defined as a Borel measure $\nu_H$ on the spectral (energy) axis. The support of this measure is the essential spectrum of the operator. Specifically, if we have a self-adjoint, not necessarily bounded operator $H$ operator on $L_2(\mathbb{R}^d)$, it is said to have a density of states if for all $g\in C_0(\mathbb{R})$ (continuous functions vanishing at infinity) the following limit exists:
\begin{equation} \label{D: DOS cont}
    \lim_{R\to\infty} \frac{1}{|B(0,R)|}\tr(g(H)M_{\chi_{B(0,R)}}),
\end{equation}
where $B(0,R) = \set{ x\in \mbR^d \colon \abs{x} \leq R},$ $\abs{\;\cdot\;}$ is the Lebesgue measure, $\tr$
is the usual operator trace, and for a function $f \colon \mbR^d \to \mbC$
we denote by $M_f$ the operator of multiplication by $f$ which acts on $L_2(\mbR^d).$ This defines a positive linear functional on $C_0(\mathbb{R})$ which therefore extends uniquely to a positive Borel measure $\nu_H$ on $\mathbb{R}$~\cite[Section~C7]{Simon1982} such that \[
    \lim_{R\to\infty} \frac{1}{|B(0,R)|}\tr(g(H)M_{\chi_{B(0,R)}}) = \int_\mathbb{R} g d\nu_H.\]

Analogously, for a countably infinite discrete metric space $(X,d_X)$ such that all balls contain finitely many points, we say for self-adjoint, not necessarily bounded operators $H$ on $\ell_2(X)$ (square-integrable functions $L_2(X, \mu)$ where we take $\mu$ the counting measure) that the density of states with respect to a fixed base-point $x_0\in X$ exists if for all $g \in C_0(\mathbb{R})$ the limit \begin{equation}\label{D: DOS disc}
    \lim_{R\to\infty} \frac{1}{|B(x_0,R)|}\tr(g(H)M_{\chi_{B(x_0,R)}})
\end{equation}
exists, where now $B(x_0,R) = \set{ x\in X \colon d_X(x, x_0) \leq R}$ and $\abs{\;\cdot\;}$ denotes the counting measure on $X$. Again this defines a positive linear functional on $C_0(\mathbb{R})$ which admits a unique extension to a Borel measure $\nu_H$ on $\Rl$~\cite[Section~C7]{Simon1982} such that \[    \lim_{R\to\infty} \frac{1}{|B(x_0,R)|}\tr(g(H)M_{\chi_{B(x_0,R)}}) = \int_\mathbb{R} g d\nu_H,\]
which we will take as the definition of the density of states of $H$. 

The present paper is a companion to~\cite{AMSZ} which covers a continuous analogue of the main result in this paper. We will first give a brief summary of this continuous case. In $\mathbb{R}^d$ it is known that the DOS has a certain stability property --- it ignores localised perturbations of the operator $H$. On the other hand, in non-commutative geometry there is a prominent tool called the Dixmier trace which has a similar property. A Dixmier trace $\tr_\omega$ is a tracial functional on the ideal $\clL_{1,\infty}$ of compact operators $A$ on a Hilbert space $\hilb$
whose singular values $\mu_n(A)$ obey $\mu_n(A) \leq C/n$ for some $C = C(A) >0$ (we refer to Section~\ref{S: Preliminaries} for a more thorough explanation). This trace is singular in the sense that it vanishes on trace class operators, and as such it shares the property of being insensitive to `small' perturbations with the density of states. Inspired by this, some of the authors found the following connection between the two~\cite{AMSZ}: 

\begin{theorem} \label{T: old case}
Let $d$ be an integer  $\geq 2$ and $M_{\langle x \rangle}$ the operator of multiplication by $(1+\abs{x}^2)^{1/2}$ on $L_2(\mbR^d).$ 
 Let $H = -\Delta + V$ be a Schr\"odinger operator on $L_2(\mbR^d)$, where~$V$ is a bounded real-valued measurable potential.
For any $g \in C_c(\mbR)$ the operator
$
g(H) M_{\langle x \rangle}^{-d}  
$
belongs to the weak trace-class ideal $\clL_{1,\infty}(L_2(\mbR^d))$. If we assume that the density of states of~$H$ (defined according to \eqref{D: DOS cont}) exists and is a Borel measure $\nu_H$ on $\mathbb{R}$, then for every Dixmier trace
$\tr_\omega$ on~$\clL_{1,\infty}$ it holds that
\begin{equation} \label{F: phi(..)=mu(g) cont case}
\tr_\omega\brs{g(H) M_{\langle x \rangle}^{-d} } = \frac{\volSd}{d}  \int_{\mbR} g\,d\nu_H,
\end{equation}
where 
$
\omega_d = {2\pi^{d/2}} / {\Gamma\left(\frac{d}{2}\right)}
$
is the $(d-1)$-volume of the unit sphere $S^{d-1}$. 
\end{theorem}

This formula, apart from providing a connection between two objects, one from non-commutative geometry and another from solid-state physics,
opens ways to explore the density of states using tools of operator theory.

In a way, we will return to the origin of the density of states in this paper and give the analogue of Theorem~\ref{T: old case} for certain discrete metric spaces that include crystals, the original source of the concept of a DOS (see~\cite{HoddesonBaym1987} and a very early use of the DOS in this context in 1929 by F. Bloch~\cite{Bloch1929}). The result, Theorem~\ref{T: Main} below, has novel aspects compared to~\cite{AMSZ}: it is applicable to a vast variety of discrete metric spaces, only constrained by a certain growth condition on its metric balls. There is no constraint on the dimension of this metric space, in contrast to the above continuous version of the theorem.

The main theorem of this paper is then the following, the proof of which is specific to the discrete case and is based on recent advances in operator theory and the notion of $V$-modulated operators hatched in the theory of singular traces, see~\cite{KaltonLord2013} or~\cite[Section~7.3]{LSZVol1}. Once more we refer to Section~\ref{S: Preliminaries} for the definition of Dixmier traces and the space $\ell_{1,\infty}(X)$.
\begin{theorem}\label{T: Main}
    Let $(X,d_X)$ be a countably infinite discrete metric space such that every ball contains at most finitely many points, and let $x_0 \in X$. Then the image of the map $d_X(\cdot, x_0): X \rightarrow \mathbb{R}_{\geq 0}$ is a collection of isolated points which can be ordered in an increasing way, denote this by $\{r_k\}_{k \in \mathbb{N}} \subseteq \mathbb{R}$. Suppose that \begin{equation}\tag{C}\label{E: Condition}
        \lim_{k\rightarrow \infty}\frac{|B(x_0, r_{k+1})|}{|B(x_0, r_k)|} = 1.
    \end{equation} Recall that $\abs{\;\cdot\;}$ indicates the counting measure on $X$. Then for any positive, radially strictly decreasing function $w \in \ell_{1,\infty}(X)$ we have that for every Dixmier trace $\tr_\omega$
    \begin{equation} \label{E: general DOS formula}
        \tr_\omega(TM_w) = \tr_\omega(M_w)\lim_{k\to \infty} \frac{\tr(TM_{\chi_{B(x_0,r_k)}})}{|B(x_0,r_k)|}
    \end{equation}
    for all bounded linear operators $T$ on $\ell_2(X)$ for which the limit on the right-hand side exists.
\end{theorem}
Because $\{r_k\}$ denotes all possible distances $d_X(x_0, \cdot)$ in the discrete space $X$, the limit on the right-hand side of equation~\eqref{E: general DOS formula} exists if and only if the continuous limit \[\lim_{R\to \infty} \frac{\tr(TM_{\chi_{B(x_0,R)}})}{|B(x_0,R)|}\] exists, and these limits are necessarily equal. This is therefore in line with the definition of the density of states~\eqref{D: DOS disc}.

In the next section, on concrete examples taken from physics, we will demonstrate that condition~\eqref{E: Condition} appearing in this theorem is very natural.

The function $w$ appearing in Theorem~\ref{T: Main} can always be chosen as follows: \begin{equation*}
    w(x) := \frac{1}{1+|B(x_0,d_X(x_0,x))|},\quad x \in X.
\end{equation*}
However, the theorem holds for any strictly positive, radially strictly decreasing $w \in \ell_{1,\infty}(X)$ (radially decreasing meaning that $w(x)$ is a decreasing function of $d_X(x_0, x)$).

Observe that if $H$ is a self-adjoint, possibly unbounded, operator on $\ell_2(X)$ with density of states measure $\nu_H$, then Theorem~\ref{T: Main} implies that for all $f \in C_c(\Rl)$ we have:
\begin{equation}\label{F: DOS integral}
    \tr_\omega(f(H)M_w) = \tr_\omega(M_w)\int_{\Rl} f\,d\nu_H.
\end{equation}

From this formula it is clear that $f(H)M_w$ is Dixmier measurable for any $f\in C_c(\mathbb{R})$ whenever $H$ admits a density of states and $M_w$ itself is Dixmier measurable as well. The converse does not hold, as the following example shows.

\begin{example}
Let $X = \mathbb{N}$ with the usual metric $d(x,y) = |x-y|$, and take the base-point $x_0 = 1$. Note that an operator $A \in \mathcal{L}_{1,\infty}$ is Dixmier measurable if and only if \[\frac{1}{\log(2+n)} \sum_{k=1}^n \lambda(k,A)\] converges~\cite[Theorem~9.1.2(c)]{LSZVol1}, where $\lambda(k,A)$ is an eigenvalue sequence of $A$ ordered such that $\abs{\lambda(k,A)}$ is decreasing. It is clear that for $w(n) := n^{-1}$, the operator $M_w$ on $\ell_2(\mathbb{N})$ is in $\mathcal{L}_{1,\infty}$, and by the above it is Dixmier measurable. Define a self-adjoint bounded operator $H$ on $\ell_2(\mathbb{N})$ by $H(e_n) = \lambda_n e_n$, where $\lambda_n$ is defined as $\lambda_1 := 1$,
\begin{equation*}
\lambda_n := 
\begin{cases}
\begin{aligned}
0, \quad&n \in [2^{2m}+1, 2^{2m+1}], &&m=0,1,2,\dots,\\
1, \quad&n \in [2^{2m-1}+1, 2^{2m}], &&m=1,2,3,\dots.
\end{aligned}
\end{cases}
\end{equation*}
\end{example}
Observe that $\frac{1}{n}\sum_{k=1}^n \lambda_n$ does not converge as $n\to \infty$, since 
\begin{align*}
\frac{1}{2^{2m}} \sum_{k=1}^{2^{2m}} \lambda_n &= \frac{1}{2^{2m}}\sum_{k=0}^{2m}(-2)^k\\
&= \frac{1}{2^{2m}}\left(\frac{1-(-2)^{2m+1}}{1+2}\right)\\
&\to \frac{2}{3};\\
\frac{1}{2^{2m+1}} \sum_{k=1}^{2^{2m+1}} \lambda_n &= \frac{1}{2^{2m+1}}\sum_{k=0}^{2m}(-2)^k\\
&= \frac{1}{2^{2m+1}}\left(\frac{1-(-2)^{2m+1}}{1+2}\right)\\
&\to \frac{1}{3},
\end{align*}
where we used the closed-form formula for the sum of a geometric series. Hence \[\frac{1}{|B(x_0,n)|}\tr(M_{\chi_{B(x_0,n)}}H) = \frac{1}{n}\sum_{k=1}^n \lambda_n\] does not converge as $n\to \infty$, and $H$ does not admit a density of states. However, $\frac{1}{\log(2+n)}\sum_{k=1}^n \frac{\lambda_k}{k}$ does converge. For convenience, we will show the convergence of $\frac{1}{\log(n)}\sum_{k=1}^n \frac{\lambda_k}{k}$, taking $n>2$, and we will use that the harmonic number $H(n)=\sum_{k=1}^n \frac{1}{k}$ has the expansion $H(n) = \log(n)+\gamma+\frac{1}{2n}+O\left(\frac{1}{n^2}\right)$, where $\gamma$ is the Euler-Mascheroni constant~\cite[Section~1.2.7]{Knuth1997}. For $m\geq 1$,
\begin{align*}
\frac{1}{\log(2^{2m})}\sum_{k=1}^{2^{2m}}\frac{\lambda_k}{k} &= \frac{1}{2m\log(2)}\left(1+\sum_{k=1}^{m} H(2^{2k})-H(2^{2k-1})\right)\\
&= \frac{1}{2m\log(2)}\sum_{k=1}^{m} \log(2) (2k-(2k-1))+\\
& \quad + \frac{1}{2m\log(2)}+\frac{1}{2m\log(2)}\sum_{k=1}^{m} \frac{1}{4k} - \frac{1}{4k-2} + O\left(\frac{1}{k^2}\right)\\
&\rightarrow \frac{1}{2},
\end{align*}
where for the last term we used that for any sequence that converges to zero, its Ces\`aro mean also converges to zero. Likewise,\begin{align*}
\frac{1}{\log(2^{2m+1})}\sum_{k=1}^{2^{2m+1}}\frac{\lambda_k}{k} &= \frac{1}{(2m+1)\log(2)}\left( 1+ \sum_{k=1}^{m} H(2^{2k})-H(2^{2k-1})\right)\\
&= \frac{m}{2m+1}+ \frac{1}{(2m+1)\log(2)}+\\
& \quad +\frac{1}{(2m+1)\log(2)}\sum_{k=1}^{m} \frac{1}{4k} - \frac{1}{4k-2} + O\left(\frac{1}{k^2}\right)\\
&\rightarrow \frac{1}{2}.
\end{align*}
Note that $\frac{1}{\log(n)}\sum_{k=1}^n \lambda_k/k$ increases monotonically on $n\in [2^{2m-1}+1, 2^{2m}]$ and decreases monotonically on $n\in[2^{2m}+1,2^{2m+1}]$. Therefore the above shows that $\frac{1}{\log(n)}\sum_{k=1}^n \lambda_k/k \rightarrow \frac{1}{2}$, and hence also $\frac{1}{\log(n)}\sum_{k=1}^n (1-\lambda_k)/k \rightarrow \frac{1}{2}$. To conclude, note that the eigenvalues of $f(H)M_w$ for $f\in C(\mathbb{R})$ are \[\{f(0)/k: \lambda_k = 0\} \cup \{f(1)/k : \lambda_k=1\} = \{f(0)(1-\lambda_k)/k+f(1)\lambda_k/k\}_{k\in \mathbb{N}}.\] These are not ordered in decreasing fashion, but this can be remedied via~\cite[Theorem 7.1.3]{LSZVol1}, see also Theorem~\ref{T: Modulated}. Since $f(H)M_w$ is $M_w$-modulated (see Section~\ref{S: Proof} or~\cite{KaltonLord2013} and~\cite[Section~7.3]{LSZVol1} for more information), this theorem implies that for $\lambda(k,f(H)M_w)$ ordered in the desired way, \[\sum_{k=1}^n\lambda(k,f(H)M_w) = \sum_{k=1}^n f(0)(1-\lambda_k)/k+\sum_{k=1}^n f(1)\lambda_k/k + O(1).\] Hence, \[\frac{1}{\log(2+n)}\sum_{k=1}^n \lambda(k,f(H)M_w) \to \frac{f(0)+f(1)}{2},\] and by~\cite[Theorem~9.1.2(c)]{LSZVol1} this means that $f(H)M_w$ is Dixmier measurable.

What we can learn from this example is that requiring $f(H)M_w$ to be Dixmier measurable for for all $f\in C_c(\mathbb{R})$ is a strictly weaker assumption on self-adjoint operators $H$ than requiring the existence of a DOS. Hence $\tr_\omega(f(H)M_w)$ can interpreted as a strictly stronger extension of the definition of the DOS on metric spaces that satisfy the conditions of Theorem~\ref{T: Main}.

However, the above example is not a Schr\"odinger operator. It is unknown whether these definitions of the DOS still differ when restricting to Schr\"odinger operators. Note that for the continuous case, where we take a Schr\"odinger operator $H=-\Delta + M_V$  on $L_2(\mathbb{R}^d)$, it is also still an open question how the existence of the DOS measure for $H$ is related to the Dixmier measurability of $f(H)M^{-d}_{\langle x \rangle}$ for all $f \in C_c(\mathbb{R})$. In~\cite{AMSZ} these are conjectured to be equivalent.

\medskip
The structure of the rest of the paper is as follows. First we will discuss condition~\eqref{E: Condition} in Section~\ref{S: Metric Condition} and give examples of spaces where it is satisfied. In Section~\ref{S: Preliminaries} we will then cover the preliminaries needed to understand the technical discussions that follow. Section~\ref{S: Proof} is the heart of the paper where we prove the main result, Theorem~\ref{T: Main}. Finally, in Appendix~\ref{S: measurability} we prove a theorem regarding the Connes' measurability of the operator which we take the Dixmier trace of in \eqref{E: general DOS formula}, and in Appendix~\ref{S: Translation} we apply the Dixmier trace formula for the DOS to provide a new proof of the equivariance under translations for the DOS on lattices.

\section{Metric Condition}
\label{S: Metric Condition}
The main theorem of this paper is applicable to countably infinite discrete metric spaces such that every ball contains at most finitely many points and that also satisfy property~\eqref{E: Condition} holds. Namely, we require that \begin{equation}\tag{C}
    \lim_{k\rightarrow \infty}\frac{|B(x_0, r_{k+1})|}{|B(x_0, r_k)|} = 1,
\end{equation} where $\{r_k\}_{k\in \mathbb{N}}$ is the increasing sequence created by ordering the set $\{d_X(x_0, y) : y \in X\}$ in increasing manner (which results in a sequence $r_k \rightarrow \infty$ since every ball in $X$ contains at most finitely many points).

First, observe that property~\eqref{E: Condition} is a condition on the so-called \textit{crystal ball sequence} $\{|B(x_0, r_k)|\}_{k \in \mathbb{N}}$ of the metric space $X$~\cite{ConwaySloane1997}, or alternatively after defining $S(x_{0}, r_k) := B(x_0, r_k) \setminus B(x_0, r_{k-1})$ it is a condition on the \textit{coordination sequence} $\{|S(x_0, r_k)|\}_{k \in \mathbb{N}}$~\cite{Brunner1979, ConwaySloane1997, OKeeffe1995}. 

To build some intuition, consider the following comment by J.E. Littlewood. Upon encountering the condition $\lim_{n\rightarrow \infty}\frac{\lambda_{n+1}}{\lambda_n} = 1$ he remarks~\cite{Littlewood1911}: ``[This condition is] satisfied when $\lambda_n$ is any function of less order than $e^{\varepsilon n}$ for all values of $\varepsilon$, which increases in a regular manner. When, however, $\lambda_n > e^{\varepsilon n}$, the theorem breaks down altogether.''. This observation is apt, indeed our restriction on the metric space $X$ is a strictly stronger assumption than sub-exponential growth of the sequence $\{|B(x_0, r_k)|\}_{k \in \mathbb{N}}$ (with respect to $k$, not $r_k$), but exactly what kind of regular growth plus subexponential growth would imply condition~\eqref{E: Condition} is hard to pin down.

There is an equivalent description of property~\eqref{E: Condition}, the proof of which can be found in a very recent preprint by F. Cipriani and J. Sauvageot~\cite[Proposition~2.9]{CiprianiSauvageot2021}. The proposition they prove is slightly different, but the given proof is immediately applicable to the following.

\begin{proposition}\label{P: CiprianiAsymptote}
    Let $(X,d_X)$ be an infinite, discrete metric space such that each ball contains at most finitely many points, choose some point $x_0\in X$ and order $\{d_X(x_0, y) : y \in X\}$ in increasing manner to define the sequence $\{r_k\}_{k\in \mathbb{N}}$. Then \[\lim_{k\rightarrow \infty}\frac{|B(x_0, r_{k+1})|}{|B(x_0, r_k)|} = 1\] if and only if \[|B(x_0, r)| \sim \varphi(r)\] for some continuous function $\varphi: \mathbb{R}_+ \rightarrow \mathbb{R}_+$.
\begin{proof}
The proof is exactly the same as in~\cite[Proposition~2.9]{CiprianiSauvageot2021} after replacing $N_L(x)$ by $|B(x_0, r)|$ and $M_k$ by $|B(x_0, r_{k})|$.
\end{proof}
\end{proposition}

\begin{remark}
    For the Cayley graph of the free group $\mathbb{F}_2$ we have $r_k = k$, $|B(x_0, k)| = 2^k$ and hence $|B(x_0, r)| = 2^{\lfloor r \rfloor} $, but \[\frac{|B(x_0, r)|}{2^r} = 2^{\lfloor r \rfloor-r}\] which does not converge as $r\rightarrow \infty$. This illustrates that $|B(x_0, r)| \sim \varphi(r)$ for some continuous function $\varphi$ is a stronger assumption that one might expect.
\end{remark}

In the same preprint, another condition is given which is sufficient for property~\eqref{E: Condition} to be satisfied~\cite[Proposition~2.8]{CiprianiSauvageot2021}.

\begin{proposition}\label{P: CondC}
    Let $X$ be a metric space as in Proposition~\ref{P: CiprianiAsymptote}. If $|B(x_0, r_k)| \sim f(k)$ (letting now $k\rightarrow \infty$ over the integers) for a function $f\in C^1(0,\infty)$ such that $\frac{f'(x)}{f(x)} \rightarrow 0$ as $x\rightarrow \infty$, then \[\lim_{k\rightarrow \infty}\frac{|B(x_0, r_{k+1})|}{|B(x_0, r_k)|} = 1.\] In particular, if $|B(x_0, r_k)|$ is a polynomial in $k$, then $\frac{|B(x_0, r_{k+1})|}{|B(x_0, r_k)|}\rightarrow 1$ as $k\rightarrow \infty$.
\begin{proof}
See~\cite[Proposition~2.8]{CiprianiSauvageot2021}
\end{proof}
\end{proposition}

To be used later on, we also postulate the following lemma.

\begin{lemma}\label{L: CondC}
    Let $X$ be a metric space as in Proposition~\ref{P: CiprianiAsymptote}. If there exist constants $C_1, C_2, d$ and $R$ such that for $r_k > R$ we have \begin{equation}\label{E: Coord seq cond1}
    C_1 k^d < |S(x_0, r_k)| < C_2 k^d,
    \end{equation} then $X$ has property~\eqref{E: Condition}.
\begin{proof}
    If $C_1 k^d < |S(x_0, r_k)| < C_2 k^d$, we can deduce that also $|B(x_0, r_k)| \geq \frac{C_1}{d+1} k^{d+1} + O(k^d)$ for $r_k > R$, and therefore \begin{equation}\label{E: Folner}
        \lim_{k\rightarrow \infty}\frac{|S(x_0, r_{k+1})|}{|B(x_0, r_k)|} = 0,
    \end{equation} which is equivalent with \[\lim_{k\rightarrow \infty}\frac{|B(x_0, r_{k+1})|}{|B(x_0, r_k)|} = 1. \qedhere\] 
\end{proof}
\end{lemma}

As a final general comment, when writing condition~\eqref{E: Condition} in the manner of Equation \eqref{E: Folner}, it vaguely resembles a type of F\o lner condition. In particular, it is reminiscent of work by Adachi and Sunada on the DOS on amenable groups where a closely related property is the subject of interest, namely Property (P) in~\cite[Proposition~1.1]{AdachiSunada1993}, also compare with~\cite[Lemma~3.2]{AdachiSunada1993}.

In the next section we describe natural examples coming from physics which satisfy condition~\eqref{E: Condition}.

\subsection{Solid matter}
Consider any kind of rigid matter whose atoms or molecules are arranged in Euclidean space in such a way that it can be described by a tiling of that space. To be precise, we mean a tiling generated by only a finite selection of different tiles, with each type of tile having a fixed arrangement of atoms within (at least 1). Any crystal can be described in this a way using only one tile by considering its underlying Bravais lattice~\cite[Chapter~4]{Ashcroft1976}, but the definition above includes quasicrystals~\cite{ShechtmanBlech1984, LevineSteinhardt1984}. For the approach of quasicrystals by tilings see for example~\cite{Hof1995, Jaric1989, Nelson1986}. Specifically,~\cite{Hof1993} establishes the existence of the integrated DOS, which is the existence of the function $\lambda \mapsto \nu_H(-\infty,\lambda)$, for every self-adjoint vertex-pattern-invariant operator on aperiodic self-similar tilings.

If we define the set $X$ of the metric space $(X,d_X)$ as the atoms or molecules of the material and impose the induced Euclidean metric, then we claim that this space has property~\eqref{E: Condition}.

\begin{proposition}
    Let $X$ be a discrete subset of $\mathbb{R}^d$ with the inherited Euclidean metric, such that $X$ can be defined by a tiling as described above. Then $(X,d_X)$ has property~\eqref{E: Condition}.
\begin{proof}
Without loss of generality, assume that the diameters of the tiles are all less than $1$. Hence $r_{k+1} \in (r_k, r_k +2]$

and therefore it suffices to show that \[\frac{|B(x_0, r_k +2) \setminus B(x_0, r_k)|}{|B(x_0, r_k)|} \xrightarrow{k \rightarrow \infty} 0.\]
Now, the number of vertices contained in $B(x_0, r_k +2) \setminus B(x_0, r_k)$ is bounded from above by some constant times $(r_k)^{d-1}$: if the smallest tile has volume $V$, and each tile contains at most $n$ atoms, and $\tilde{B}(x_0, r_k +2)$ denotes the ball in $\mathbb{R}^d$, then there can be at most $n \frac{|\tilde{B}(x_0, r_k +2) \setminus \tilde{B}(x_0, r_k)|}{V}$ vertices in $B(x_0, r_k +2) \setminus B(x_0, r_k)$, which is bounded by $C_1(r_k)^{d-1}$. 

If the volume of the biggest tile is $W$, the number of tiles that are \textit{fully} contained in $B(x_0, r_k)$ is similarly bounded from below by $\frac{|\tilde{B}(x_0, r_k - 1)|}{W}$ because we assumed that the diameter of the tiles is less than 1. Recall that we assumed that each tile contains at least one atom. Then the number of atoms in $B(x_0, r_k)$ can be bounded from below by $\frac{|\tilde{B}(x_0, r_k - 1)|}{W}$, which is of the form $C_2 (r_k)^d + O((r_k)^{d-1})$. Hence indeed \[0 \leq \frac{|B(x_0, r_k +2) \setminus B(x_0, r_k)|}{|B(x_0, r_k)|} \leq \frac{C_1(r_k)^{d-1}}{C_2 (r_k)^{d} + O((r_k)^{d-1})} \xrightarrow{k \rightarrow \infty} 0,\] and we see that this metric space satisfies condition~\eqref{E: Condition}.
\end{proof}
\end{proposition}
\subsection{Crystals}\label{SS: Crystals}
In another approach, we can take the atoms of crystals as the vertices of a graph and define our discrete metric space $(X,d_X)$ as this graph with shortest-path metric. Common choices for such a construction are the contact graph~\cite{ConwaySloane1997} and the Voronoi graph~\cite[p.~33]{ConwaySloane1993}. See for example the very recent paper~\cite{PradodeOliveira2021} based on the model~\cite{deOliveiraPrado2005} which shows the existence of the DOS measure on $\mathbb{Z}$ for a suitable family of Dirac operators, the very recent~\cite{PapaefstathiouRobaina2021} or~\cite{Pastur1980} for an older result.

As a first observation, note that for any graph with the shortest-path metric, we have $r_k = k$.

For crystals specifically, such a graph $\Gamma$ comes with a free $\mathbb{Z}^d$ action such that the quotient graph $\Gamma/\mathbb{Z}^d$ is finite. This is a simple observation by considering the underlying Bravais lattice of any crystal~\cite[Chapter~4]{Ashcroft1976}. A very recent advancement by Y. Nakamura,  R. Sakamoto, T. Masea and J. Nakagawa~\cite{NakamuraSakamoto2021} concerns exactly such (even possibly directed) graphs $\Gamma$ with a free $\mathbb{Z}^d$ action such that $\Gamma/\mathbb{Z}^d$ is finite. Namely, these authors have proven that the coordination sequence $\{|S(x_0, k)|\}_{k \in \mathbb{N}}$ is then of quasi-polynomial type, by which they mean the following. A quasi-polynomial is defined as a function $p: \mathbb{Z}_{\geq 0} \rightarrow \mathbb{Z}$ with $p(k) = c_m(k) k^m + c_{m-1} k^{m-1}+ \cdots + c_0(k)$ where $c_i(k)$ are all periodic with an integral period. Equivalently, it is a function such that for some integer $N$ \[p(k) = \begin{cases} p_0(k) & k = 0 \text{ mod } N\\
p_1(k) & k = 1 \text{ mod } N\\
&\vdots\\
p_{N-1}(k) & k = N-1 \text{ mod } N,\\
\end{cases}\]
where $p_0, \dots, p_{N-1}$ are polynomials. A function of quasi-polynomial type is then defined as a function $f: \mathbb{Z}_{\geq 0} \rightarrow \mathbb{Z}$ such that there exists some quasi-polynomial $p$ and an integer $M$ such that $f(n) = p(n)$ for all $n \geq M$. See also~\cite[Section~4.4]{Stanley1986}.

\begin{proposition}
Let $\Gamma$ be a graph with a free $\mathbb{Z}^d$ action such that the quotient graph $\Gamma/\mathbb{Z}^d$ is finite. Then $(\Gamma, d_\Gamma)$, where we take $d_\Gamma$ as the shortest-path metric, has property~\eqref{E: Condition}.
\end{proposition}
\begin{proof} 
By~\cite[Theorem~1.1]{NakamuraSakamoto2021} we have that $|S(x_0, k)|$ is a quasi-polynomial, denote this quasi-polynomial by $p(k)$ and its constituent polynomials by $p_1, \dots, p_{N-1}$. Suppose that the polynomial $p_r$ is one with maximal degree, i.e. $\deg(p_i) \leq \deg(p_r) = t$ for all $i=0, \dots, N-1$. Then there exist some constants $C_1, C_2$ and $L$ such that $p(k) \leq C_1 k^t$ for all $k \geq L$, and also $p_r(k) = C_2 k^t + O(k^{t-1})$. Now take $k \in \mathbb{Z}_{\geq 0}$ arbitrary, and define $a \in \mathbb{Z}_{\geq 0}$ as the smallest positive integer such that $k = mN+r +a$ for some $m \in \mathbb{Z}$. It follows that $a \in \{0, \dots, N-1\}$, and hence note that $mN = k-r-a \geq k-2N$, i.e. $\frac{k}{N} -2\leq m \leq \frac{k}{N}$.
\begin{align*}
    \sum_{n=1}^k p(n) & \geq \sum_{n=0}^m p_r(nN + r)\\
    & = \sum_{n=0}^m C_2(nN + r)^t + O((m+1) (nN + r)^{t-1})\\
    & = C_2\sum_{n=0}^m (nN)^d + O(k^{t})\\
    & = C_3 k^{t+1} + O(k^t).
\end{align*}
Therefore we have that \[0 \leq \frac{p(k+1)}{\sum_{n=1}^{k} p(n)} \leq \frac{C_1 (k+1)^t}{C_3 k^{t+1} + O(k^t)} \xrightarrow{k \rightarrow \infty} 0,\] i.e. $\frac{|S(x_0, k+1)|}{|B(x_0, k)|}$ converges to zero as $k\rightarrow \infty$.
\end{proof}

\subsection{The integer lattice}
From the previous two subsections it follows that respectively $(\mathbb{Z}^d, \norm{\cdot}_2)$ and $(\mathbb{Z}^d, \norm{\cdot}_1)$ have property~\eqref{E: Condition}, where we define \[\norm{v}_p := \left(\sum_{i=1}^d |v_i|^p\right)^{1/p}\]for $1 \leq p < \infty$ and for $p=\infty$ \[\norm{v}_\infty:= \sup_{i=1, \dots, d} \abs{v_i}.\]
Even though the base space of these two metric spaces is the same, the difference between these is the domain undergoing the thermodynamic limit in the definition of the DOS, which can make a difference as demonstrated in the preprint~\cite{AMZS_domains}. As mentioned, the existence of the DOS in the case $\mathbb{Z}$ has been established for a suitable family of Dirac operators~\cite{deOliveiraPrado2005}. Another example is that, using $(\mathbb{Z}^d, \norm{\;\cdot\;}_\infty)$ as a model, the existence of a surface DOS was established for a quantum model with a surface~\cite{EnglischKirsch1988}.

In fact, $(\mathbb{Z}^d, \norm{\cdot}_p)$ has property~\eqref{E: Condition} for all $1 \leq p \leq \infty$. This is simply because in the metric space $(\mathbb{Z}^d, \norm{\cdot}_p)$ we have $\abs{B(0,r)} = V_p(d) r^d + O(r^{d-1})$ where $V_p(d)$ denotes the volume of the $\ell_p$ unit ball in $\mathbb{R}^d$, and Proposition~\ref{P: CiprianiAsymptote} then implies condition~\eqref{E: Condition}. 

Since much of the mathematical research on the DOS on crystals is still restricted to $\mathbb{Z}^d$~\cite{HislopMuller2008, Pastur1980, Wegner1981} we will explicitly demonstrate the calculation of the factor $\tr_\omega(M_w)$ appearing in Theorem~\ref{T: Main}. The choice of $w(v) = (1+\|v\|_{\ell_p})^{-d}$ is convenient for calculations, although $\tilde{w}(v) = (1+|B(0, \|v\|_{\ell_p})|)^{-1}$ would work too as will be demonstrated in Lemma~\ref{L: w in general}. 

\begin{proposition}
Define $w(v) = (1+\|v\|_{\ell_p})^{-d}$. Then $w \in \ell_{1,\infty}(\Itgr^d)$, and $\tr_\omega(M_w) = V_p(d)$.
\begin{proof}
Observe that  \begin{align*}
    \lvert\{v\in \mathbb{Z}^d: (1+\norm{v}_p)^{-d} \geq t\}\rvert &= \lvert\{v \in \mathbb{Z}^d: \norm{v}_p \leq t^{-1/d} -1\}\rvert\\
    &= \lvert B(0,t^{-1/d}-1)\rvert \\
    &= O(t^{-1}),
\end{align*} and hence $w \in \ell_{1,\infty}(\Itgr^d)$ by~\eqref{F: w in weak l1 def}.

By a straightforward variant of Lemma~\ref{L: subsequence summation}, for a positive sequence $\{a_k\}_{k \in \mathbb{N}}$ and indices $\{k_1, k_2, ...\}$ such that $\frac{\log(2+k_{i+1})}{\log(2+k_i)} \to 1$ we have that $\frac{1}{\log(2+n)}\sum_{k=1}^n a_k$ converges if and only if $\frac{1}{\log(2+k_i)} \sum_{k=1}^{k_i} a_k$ does. Since clearly \[\frac{\log(2+|B(0,k+1)|)}{\log(2+|B(0,k)|)}\] converges to 1, to calculate $\frac{1}{\log(2+n)} \sum_{k =1}^n \mu(k,M_w)$ it suffices to compute \[\lim_{k \rightarrow \infty} \frac{1}{\log(2+\abs{B(0,k)})} \sum_{x \in B(0,k)} (1+\norm{x}_p)^{-d}.\]

Since on each shell $B(0,k) \setminus B(0,k-1)$ we have that $k-1 \leq \norm{x}_p \leq k$ by definition, it follows that
\begin{align*}
   \sum_{j=0}^k j^{-d} (\abs{B(0,j)} - \abs{B(0,j-1)}) &\leq \sum_{\norm{x}_p \leq k} (1+ \norm{x}_p)^{-d}\\
   &\leq \sum_{j=0}^k (1+j)^{-d} (\abs{B(0,j)} - \abs{B(0,j-1)}).
\end{align*}

First tackling the upper estimate note that by summation by parts
\begin{align*}
    \sum_{j=0}^k (1+j)^{-d} (\abs{B(0,j)} - \abs{B(0,j-1)}) &= \frac{\abs{B(0,k)}}{(1+k)^d} + \sum_{j=0}^{k-1} \abs{B(0,j)} ((1+j)^{-d} - (2+j)^{-d}).
\end{align*}
As \[\lim_{k\rightarrow \infty} \frac{1}{\log(2+\abs{B(0,k)})} \frac{\abs{B(0,k)}}{(1+k)^d} = 0,\]
let us calculate \begin{align*}
    \lim_{k \rightarrow \infty} \frac{1}{\log(2+ V_p(d) k^d)} &\sum_{j=0}^{k-1} V_p(d) j^d ((1+j)^{-d} - (2+j)^{-d}) \\
    = \lim_{k \rightarrow \infty} \frac{1}{d\log(k)} &\sum_{j=0}^{k-1} V_p(d) j^d ((1+j)^{-d} - (2+j)^{-d}).
\end{align*} 
Note that \begin{align*}
    \frac{j^d}{(1+j)^d} - \frac{j^d}{(2+j)^d} &= \frac{(j+2)^d-(j+1)^d}{(j+3+2/j)^d} \\
    &= \frac{d (j+1)^{d-1}}{(j+3+2/j)^d} + \frac{O(j+1)^{d-2}}{(j+3+2/j)^d},
\end{align*}
and hence \begin{align*}
    \lim_{k \rightarrow \infty} \frac{1}{d\log(k)} \sum_{j=0}^{k-1} V_p(d) j^d ((1+j)^{-d} - (2+j)^{-d}) &=V_p(d) \lim_{k \rightarrow \infty} \frac{1}{\log(k)} \sum_{j=0}^{k-1} \frac{ (j+1)^{d-1}}{(j+3+2/j)^d}\\
    &= V_p(d),
\end{align*}
implying that $\tr_\omega(M_w) \leq V_p(d)$. Following exactly the same steps for the lower estimate, we arrive at the conclusion that indeed $\tr_\omega(M_w) = V_p(d)$.
\end{proof}
\end{proposition}

\subsection{Quasicrystals}
Analogously to Subsection~\ref{SS: Crystals} one can consider graphs constructed from quasicrystals, but these will be aperiodic by definition~\cite{Jaric1989, Nelson1986}. For some investigations of the DOS on quasicrystals, see~\cite{Hof1993, Hof1995}.

On a case-by-case basis, there are some aperiodic tilings for which condition~\eqref{E: Condition} can be expected to hold on their vertex graph. Firstly, a Penrose tiling. This proposition is entirely based on recent results by A. Shutov and A. Maleev~\cite{ShutovMaleev2015, ShutovMaleev2018}.

\begin{proposition}
    Consider a 2D Penrose tiling in the construction of~\cite{ShutovMaleev2015}, which is a Penrose tiling with five-fold symmetry with respect to a chosen origin 0. Consider the graph induced by this tiling (i.e. with the same vertices and edges as the tiles) and take this graph with the shortest-path metric as the definition of the metric space $(X,d_X)$. Then this metric space has property~\eqref{E: Condition}.
\begin{proof}
For this particular tiling, Shutov and Maleev showed~\cite{ShutovMaleev2018} that \[|S(0, k)| = C(n)n + o(n)\] where $C(n)$, denoting $\tau = (1+\sqrt{5})/2$, takes a value between $10 \tau^{-2} \approx 3.8$ and $10 \tau^{-2} + (5/2)\tau^{-1} \approx 5.4$ depending on $n$. By Lemma~\ref{L: CondC}, we can then conclude that the vertex graph of this Penrose tiling satisfies condition~\eqref{E: Condition}. 
\end{proof}
\end{proposition}

\begin{remark}
Similar asymptotic behaviour can be expected when taking a different base point, or different Penrose tilings, but this is still an open problem. Numerical data supports this conjecture for Penrose tilings~\cite{BaakeGrimm2006}, as well as that this metric condition will be satisfied by aperiodic tilings like the Ammann-Beenker tiling~\cite{BaakeGrimm2006}, a certain class of quasi-periodic self-similar tilings~\cite{ShutovMaleev2010} and two-dimensional quasiperiodic Ito–Ohtsuki tilings~\cite{ShutovMaleev2008}. To the authors' knowledge, a quasiperiodic tiling has not yet been found that can serve as a counter-example to property~\eqref{E: Condition}.
\end{remark}

\subsection{Percolation}
A successful model for studying conduction properties of a crystal with impurities via the density of states is percolation (in this context also called quantum percolation)~\cite{AlexanderOrbach1982, ChayesChayes1986, KirschMuller2006, Veselic2005, YakuboNakayama1987}, see also~\cite[Section~13.2]{Grimmett1993} for a general approach. In broadest generality, percolation describes the study of a statistical procedure on a graph, which means adding, removing or otherwise manipulating edges or vertices based on some probabilistic method~\cite[Chapter~1]{Grimmett1993}.

Let us focus on so-called bond percolation on the graph $\mathbb{Z}^d$. Choosing some chance $0 \leq p \leq 1$, we declare each edge on the graph $\mathbb{Z}^d$ to be \textit{open} with probability $p$ (in an independent manner), and closed with probability $1-p$. It is a well-known fact that there exists a phase transition at a critical probability $p_c(d)$~\cite[Chapter~1]{Grimmett1993}. Namely, for $p>p_c(d)$ there exists almost surely a unique infinite cluster of vertices connected by open edges, while for $p<p_c(d$ almost surely all clusters of vertices that are connected by open edges are finite.

In the super-critical region, meaning $p > p_c(d)$, one can wonder if our metric condition holds on the infinite cluster, meaning that we take $(X,d_X)$ as the infinite cluster with induced shortest-path metric. This turns out to be true, which follows quite directly from a recent result by R. Cerf and M. Th\'eret~\cite{CerfTheret2016}.

\begin{proposition}
    Let $X$ be the (almost surely unique) infinite cluster on $\mathbb{Z}^d$ after super-critical bond percolation (with $p > p_c(d)$), and let $d_X$ be the shortest path metric on $X$. Then $(X,d_X)$ has property~\eqref{E: Condition}.
\begin{proof}
We first set up the exact situation as in~\cite{CerfTheret2016}. Denote the set of edges in our graph $\mathbb{Z}^d$ by $\mathbb{E}^d$, and consider a family of i.i.d. random variables $(t(e), e\in \mathbb{E}^d)$ taking values in $[0, \infty]$ (including $\infty$), with common distribution $F$. To be very precise, for each variable $t(e)$ we take $[0,\infty]$ as sample space, we define a $\sigma$-algebra by declaring $A \subseteq [0,\infty]$ measurable if $A \setminus \{\infty\}$ is Lebesgue measurable in $\mathbb{R}$, and $F: [0,\infty] \rightarrow [0,1]$ is a measurable function that defines the distribution of the variable $t(e)$.

These variables can be interpreted as the time it takes to travel along the corresponding edge. Then consider the random extended metric on $\mathbb{Z}^d$ by defining for $x, y\in \mathbb{Z}^d$ \[T(x, y) = \inf\{\sum_{e \in \gamma} t(e) : \gamma \text{ is a path from } x \text{ to } y \}.\] Note that if the distribution $F$ is such that $F(\{1\}) = p$, $F(\{\infty\}) = 1-p$, $T(x,y)$ is always either a positive integer or infinite and defines precisely the usual induced shortest-path metric on $\mathbb{Z}^d$ after bond percolating on $\mathbb{Z}^d$ with chance $p$. Also observe that generally, if $F([0,\infty)) > p_c(d)$, there exists almost surely a unique infinite connected cluster of vertices~\cite{CerfTheret2016}.

Now define $B(x_0, t) := \{y \in \mathbb{Z}^d : T(x_0 ,y)\leq t\}$. Then by~\cite[Theorem~5(ii)]{CerfTheret2016}, if $F([0,\infty)) > p_c(d)$, $F(\{0\}) < p_c(d)$ and $x_0$ is a vertex on the infinite cluster, as $t\rightarrow \infty$ we have almost surely \[\frac{|B(x_0, t)|}{t^d} \rightarrow C\] for some constant $C \in \mathbb{R}$. Returning to the distribution $F(\{1\}) = p$, $F(\{\infty\}) = 1-p$, we can then conclude that \[\frac{|B(x_0, k+1)|}{|B(x_0, k)|} \frac{k^d}{(k+1)^d} \rightarrow 1 \] as $k\rightarrow \infty$, and hence \[\frac{|B(x_0, k+1)|}{|B(x_0, k)|} \rightarrow 1. \qedhere\] \end{proof}
\end{proposition}

\begin{remark}
Site percolation is similar to bond percolation, but as the name suggests we would assign each vertex to be open with probability $p$ and closed with probability $1-p$~\cite[Section~1.6]{Grimmett1993}. Arguably, this is a more physically suitable model of an alloy. The argument used above to prove the cited result for bond percolation can be used for site percolation as well, but the authors are not aware of an explicit demonstration.
\end{remark}

\section{Preliminaries}\label{S: Preliminaries}
In this exposition of the preliminary material we follow~\cite{AMSZ}.
This material is standard and can be found in~\cite{LSZVol1,Simon2005}.
In this paper by a Hilbert space we mean a complex separable Hilbert space which we will denote by~$\hilb.$ 
Let $\clB(\clH)$ be the set of all bounded operators on $\clH$, and $\clK(\clH)$  the ideal of compact operators on~$\clH$. For $T\in \clK(\clH)$ its singular values $\mu(T) = \{\mu(k,T)\}_{k=0}^\infty$ are defined as:
\begin{equation*}
    \mu(k,T) = \inf\{\|T-R\|\;:\;\mathrm{rank}(R) \leq k\}.
\end{equation*}
Let $p \in (0,\infty).$ The Schatten class $\clL_p(\hilb)$ is the set of compact operators $T$ on $\hilb$ such that the sequence $\mu(T)$ is $p$-summable. For $p \geq 1,$ the $\clL_p$
norm is defined as:
\begin{equation*}
    \|T\|_p := \|\mu(T)\|_{\ell_p} = \left(\sum_{k=0}^\infty \mu(k,T)^p\right)^{1/p},
\end{equation*}
which turns $\clL_p$ into a Banach space and an ideal of $\clB(\clH)$.

The weak Schatten class $\clL_{p,\infty}$ is the set of compact operators $T$ such that $\mu(T)$ obeys:
\begin{equation*}
    \|T\|_{p,\infty} := \sup_{k\geq 0} (k+1)^{1/p}\mu(k,T) < \infty.
\end{equation*}
$\|T\|_{p,\infty}$ is a quasi-norm. 
The sequence space $\ell_{p,\infty}$ is the set of all sequences of complex numbers $a = \set{a_k}_{k=0}^\infty$ converging to zero such that 
$$
    \|a\|_{p,\infty} := \sup_{k\geq 0} (k+1)^{1/p} a^*_k < \infty,
$$
where $a^*_k$ is a non-increasing rearrangement of $\abs{a_k}.$
The set $\ell_{p,\infty}$ is called the weak $L_p$-space. A compact operator $T$ on $\hilb$ belongs to $\clL_{p,\infty}$ if and only if $\mu(T)$ belongs to $\ell_{p,\infty},$
in which case $\|T\|_{p,\infty} =     \|\mu(T)\|_{p,\infty}.$

Let $X$ be a countable set. 
The space $\ell_{1,\infty}(X)$ of functions consists
of those functions $w \colon X \to \Cplx$ 
such that the operator $M_w$ of multiplication by $w$ belongs to $\clL_{1,\infty}(\ell_2(X)).$
Equivalently, $w \in \ell_{1,\infty}(X)$ iff there exists $C>0$ such that for any $t>0$
\begin{equation} \label{F: w in weak l1 def}
   | \{v \in X \colon |w(v)| \geq t\}| \leq C t^{-1}.
\end{equation}

A linear functional $\varphi: \clL_{1,\infty} \rightarrow \mathbb{C}$ is called a trace if satisfies the property that $\varphi(BT ) = \varphi(TB)$ for all $T\in \clL_{1,\infty}$ and all bounded operators $B$. This condition is equivalent with requiring that $\varphi(U^*TU) = \varphi(T)$ for all $T \in \clL_{1,\infty}$ and all unitary operators $U$. 

On $\clL_{1,\infty}(\hilb)$ a particular kind of traces can be defined that are called Dixmier traces, see for example~\cite[Chapter~6]{LSZVol1}. A Dixmier trace is a trace on $\clL_{1,\infty}$ that is based on an extended limit $\omega \in \ell_\infty(\mathbb{N})^*$, by which we mean that $\omega$ is a continuous functional that extends the usual limit functional. They are defined by the formula \[\tr_\omega: T \mapsto  \omega\left(\left\{\frac{1}{\log(2+N)}\sum_{k=0}^N \lambda(k,T)  \right\}_{N=0}^\infty\right), \quad T\in \clL_{1,\infty},\] where $\{\lambda(k,T)\}$ is any eigenvalue sequence of $T$, ordered such that $\{|\lambda(k,T)|\}$ is non-increasing~\cite[Section~6.1]{LSZVol1}. Originally, dilation-invariance was required of the extended limit $\omega$, meaning that $\omega \circ \sigma_n = \omega$ for all $n\geq 1$ where $\sigma_n(\{a_j\}_{j=0}^\infty) = \{a_{\lfloor \frac{j}{n}\rfloor}\}_{j=0}^\infty$, but this assumption is redundant~\cite[Theorem~6.1.3]{LSZVol1}.

\section{Dixmier trace formula for the DOS}\label{S: Proof}
In this section we prove Theorem~\ref{T: Main} based on a series of lemmas. The cause of the restriction on the metric space $X$ as discussed in Section~\ref{S: Metric Condition} can be found in Lemma~\ref{L: subsequence summation}.

To start off, we will give a lemma that demonstrates the existence of a function $w$ such that $M_w$ has positive Dixmier trace as required by Theorem~\ref{T: Main}.

\begin{lemma} \label{L: w in general}
    Let $(X,d_X)$ be a countably infinite discrete metric space with the property that any ball contains finitely many points, and let $x_0 \in X$. Then the function
        $w(v) = (1+|B(x_0,d_X(v,x_0))|)^{-1}$
    satisfies the conditions of Theorem~\ref{T: Main}.
\end{lemma}
\begin{proof}
    The assumption that any ball contains finitely many points ensures that $w > 0$. To see that $w \in \ell_{1,\infty}(X)$, note that for all $t > 0$ we have:
    \begin{align*}
        |\{v \in X\;:\; |w(v)| \geq t\}| &= |\{v \in X\;:\; 1+|B(x_0,d_X(v,x_0))| \leq t^{-1}\}|\\
                                         &\leq |\{v\in X\;:\; |B(x_0,d_X(v,x_0))| < t^{-1}\}|.
    \end{align*} 
    Therefore if $t^{-1} = |B(x_0,R)|$, we have:
    \begin{align*}
        |\{v\in X\;:\; |w(v)| \geq t\}| &\leq |\{v \in X\;:\; |B(x_0,d_X(v,x_0))| < |B(x_0,R)|\}|\\
                                        &\leq |B(x_0,R)| 
                                        = t^{-1}.
    \end{align*}
    Since $|X|$ is infinite, $t$ can be arbitrarily small and hence $w \in \ell_{1,\infty}(X)$. \qedhere
\end{proof}

We now present a modified Toeplitz lemma, which follows from much the same proof as in Shiryaev~\cite[Chapter IV, \S 3, Lemma 1]{Shiryaev1996}.

\begin{lemma}[Toeplitz lemma]\label{toeplitz_lemma}
    Let $\{c_n\}_{n=0}^\infty$ be a sequence of non-negative numbers, and let $\{z_n\}_{n=0}^\infty$ be a sequence of complex numbers such that $z_n\rightarrow L \in \Cplx$. 
    If $d_n = \sum_{k=0}^n c_k$ diverges, then:
    \begin{equation*}
        \sum_{k=0}^n c_k z_k = L d_n + o(d_n).
    \end{equation*} 
    as $n\to \infty$.
\end{lemma}
\begin{proof}
    Let $\varepsilon > 0$ and take $K>0$ sufficiently large such that if $k > K$ then $|z_k - L| < \varepsilon$. For any $n> K,$ 
    rewriting the left hand side of the equality above as
    \begin{align*}
        \sum_{k=0}^n c_k z_k &= \sum_{k=0}^n c_k L 
                           +\sum_{k=0}^K c_k(z_k - L) + \sum_{k=K+1}^n c_k(z_k - L),
    \end{align*}
    we see that 
    \begin{equation*}
        \left|\frac{1}{d_n}\sum_{k=0}^n c_k(z_k - L) \right| \leq \frac{1}{d_n}\sum_{k=0}^K c_k|z_k - L| + \varepsilon.
    \end{equation*}
    Since $d_n\to \infty$ as $n\to \infty,$ it follows that:
    \begin{equation*}
        \limsup_{n\to \infty}\left|\frac{1}{d_n} \sum_{k=0}^n c_k(z_k - L) \right| \leq \varepsilon.
    \end{equation*}
    Since $\varepsilon$ is arbitrary, we have:
    \begin{equation*}
        \sum_{k=0}^n c_k(z_k - L)  = o(d_n)
    \end{equation*}
    and this completes the proof.
\end{proof}

\begin{lemma}\label{toeplitz_corollary}
    Let $\{x_k\}_{k=0}^\infty$ and $L \in \Cplx$ be such that as $n\to\infty$,
    \begin{equation*}
        \sum_{k=0}^n x_k = L n+o(n).
    \end{equation*}
    Let $\{a_n\}_{n=0}^\infty$ be a sequence of non-negative numbers such that:
    \begin{enumerate}[{\rm (i)}]
        \item{} $\{a_n\}_{n=0}^\infty$ is non-increasing,
        \item{} $\sup_{k\geq 1} ka_k < \infty$. That is, the sequence is weak $\ell_1$,
        \item{} $b_n := \sum_{k=0}^n a_k$ diverges.
    \end{enumerate}
    Then   
    \begin{equation*}
        \sum_{k=0}^n a_kx_k = L b_n+o(b_n)
    \end{equation*}
    as $n\to \infty$.
\end{lemma}
\begin{proof}
    Let $y_n = \sum_{k=0}^{n-1} x_k$ with $y_{0}=0$. Abel's summation formula gives
    \begin{align*}
        \sum_{k=0}^n a_kx_k &= \sum_{k=0}^n a_k(y_{k+1}-y_{k})\\
                            &= y_{n+1}a_n-\sum_{k=1}^n (a_k-a_{k-1})y_{k}\\
                            &= y_{n+1}a_n+\sum_{k=1}^n \frac{y_k}{k}\cdot k(a_{k-1}-a_{k}).
    \end{align*}
    By assumption (i) we have $a_{k-1} \geq a_k$ so the sequence $c_k := k(a_{k-1}-a_{k})$ is non-negative, and moreover as $n\to\infty$,
    \begin{equation*}
        \sum_{k=1}^n c_k = b_{n-1}-na_n\rightarrow \infty,
    \end{equation*}
    since by assumption (iii) $b_n \to \infty$ and by assumption (ii) $na_n$ is bounded.
    Therefore Lemma~\ref{toeplitz_lemma} applies to $c_k$ and $z_k:= \frac{y_k}{k}$, since by assumption $\lim_{k\to\infty} \frac{y_k}{k} = L$, and hence it follows that
    \begin{equation*}
        \sum_{k=1}^n \frac{y_k}{k}\cdot k(a_{k-1}-a_k) = L (b_{n-1}-na_n)+o(b_{n-1}-na_n)
    \end{equation*} 
    as $n\to\infty$.
    Thus,
    \begin{align*}
        \sum_{k=0}^n a_kx_k  & = \frac{y_{n+1}}{n}\cdot na_n + L(b_{n-1}-na_n)+o(b_{n-1}-na_n)  \\
           & = L b_{n-1}+o(b_{n-1}-na_n),
    \end{align*}
    where in the last equality we have absorbed the vanishing term $\brs{\frac{y_{n+1}}{n} - L} na_n$ into $o(b_{n-1}-na_n).$ 
    Since by assumption (ii) the sequence $\{na_n\}_{n=1}^\infty$ is bounded, it follows that:
    \begin{equation*}
        \sum_{k=0}^n a_kx_k = L b_n+o(b_n).\qedhere
    \end{equation*}
\end{proof}

We could also write the result of Lemma~\ref{toeplitz_corollary} as:
\begin{equation*}
    \lim_{n\to \infty} \frac{\sum_{k=0}^n x_k}{\sum_{k=0}^n 1} = \lim_{n\to \infty} \frac{\sum_{k=0}^n a_kx_k}{\sum_{k=0}^n a_k}
\end{equation*}
whenever the left hand side exists and $\{a_k\}_{k=0}^\infty$ satisfies the stated assumptions.

\bigskip
Our next aim is to prove Lemma~\ref{expectation_values_lemma}, which is the crux of the proof of Theorem~\ref{T: Main}.
The proof is based on the notion of a $V$-modulated operator from~\cite{KaltonLord2013} or \cite[Section 7.3]{LSZVol1}. 
\begin{definition} \label{D: V-modulated}
If $V$ is a positive bounded linear operator
on a Hilbert space, and~$T$ is a bounded linear operator on the same space, we say that $T$ is \emph{$V$-modulated} if:
\begin{equation}   \label{def V-modulated}
    \sup_{t > 0} t^{1/2}\|T(1+tV)^{-1}\|_2 < \infty.
\end{equation}
\end{definition}
As can be seen from the definition, 
a $V$-modulated operator is necessarily Hilbert-Schmidt.
The importance of $V$-modulated operators comes from the following theorem, see~\cite[Theorem 7.1.3]{LSZVol1}.
\begin{theorem}\label{T: Modulated}
If $V \in \Lc_{1,\infty}(\hilb)$ is strictly positive, $T$ is $V$-modulated and $\{e_k\}_{k=0}^\infty$ is an orthonormal basis
such that $Ve_k = \mu(k,V)e_k$, then as $n\to \infty$,
\begin{equation}\label{original_expectation_values}
    \sum_{k=0}^n \lambda(k,T) = \sum_{k=0}^n \langle e_k,Te_k\rangle +O(1),
\end{equation}
where  $\{\lambda(k,T)\}_{k=0}^\infty$ are the eigenvalues of $T$ ordered with non-increasing absolute value. 
\end{theorem}

Note that if $V \in \Lc_{1,\infty}(\hilb)$, then $V$ is automatically $V$-modulated~\cite[Lemma 7.3.4]{LSZVol1}, and if $A$ is bounded and $T$ is $V$-modulated, then $AT$ is $V$-modulated,
which directly follows from \eqref{def V-modulated}.
Changing notation slightly, it follows that if $0< W \in \Lc_{1,\infty}(\hilb)$ and $T$ is bounded, then $TW$ is $W$-modulated.
Applying \eqref{original_expectation_values} to this special case, we have as $n\to\infty$,
\begin{equation}\label{expectation_values_sum}
    \sum_{k=0}^n \lambda(k,TW) = \sum_{k=0}^n \langle e_k,Te_k\rangle \mu(k,W) + O(1)
\end{equation}
for all bounded operators $T$.

\begin{lemma}\label{expectation_values_lemma}
    Let $0 < W \in \Lc_{1,\infty}(\hilb) \setminus \Lc_1(\hilb),$ and let $\{e_k\}_{k=0}^\infty$ be an orthonormal basis such that $We_k= \mu(k,W)e_k$ for all $k\geq 0$. If $T$ is a bounded operator such that
    \begin{equation*}
        \sum_{k=0}^n \langle e_k, Te_k\rangle = Ln+o(n),\quad n\to\infty
    \end{equation*}
    where $L \in \Cplx$
    then:
    \begin{equation} \label{F: the first assertion}
        \sum_{k=0}^n \lambda(k,TW) = L\left(\sum_{k=0}^n \mu(k,W)\right)+o\left(\sum_{k=0}^n \mu(k,W)\right).
    \end{equation}
    Moreover,
    \begin{equation*}
        \tr_\omega(TW) = \tr_\omega(W)L
    \end{equation*}
    for all extended limits $\omega$.
\end{lemma}
\begin{proof}
By the assumption $0 < W \in \Lc_{1,\infty} \setminus \Lc_1,$
    Lemma~\ref{toeplitz_corollary} applies with $a_k = \mu(k,W)$ and $x_k = \langle e_k,Te_k\rangle$, so
    \begin{equation*}
        \sum_{k=0}^n \langle e_k,Te_k\rangle\mu(k,W) = L\left(\sum_{k=0}^n \mu(k,W)\right)+o\left(\sum_{k=0}^n \mu(k,W)\right),\, n\to\infty.
    \end{equation*}
    Thus \eqref{expectation_values_sum} yields the first equality \eqref{F: the first assertion}.

    To obtain the result concerning Dixmier traces, we divide both sides of \eqref{F: the first assertion} by $\log(n+2)$ to get:
    \begin{equation*}
        \frac{1}{\log(n+2)}\sum_{k=0}^n \lambda(k,TW)  = L\left(\frac{1}{\log(n+2)}\sum_{k=0}^n \mu(k,W)\right)+o(1).
    \end{equation*}
    Thus if $\omega$ is an extended limit,
    \begin{equation*}
        \omega\left(\left\{\frac{1}{\log(n+2)}\sum_{k=0}^n \lambda(k,TW)\right\}_{n=0}^\infty\right) = L\tr_\omega(W).
    \end{equation*}
    The left hand side is exactly the Dixmier trace $\tr_\omega(TW)$.
\end{proof}

\begin{remark}
The result of Lemma~\ref{expectation_values_lemma} can be written in a different way. 
We could say that:
\begin{equation}\label{feq}
    \tr_\omega(TW) = \tr_\omega(W) \lim_{n\to \infty} \frac{1}{n+1}\sum_{k=0}^n \langle e_k,Te_k\rangle
\end{equation}
whenever the right hand side exists.
\end{remark}

We now claim that
\begin{equation}\label{seq}
    \tr_\omega(TW) = \tr_\omega(W) \lim_{n\to \infty} \frac{1}{|\{k\;:\;\mu(k,W)\geq \varepsilon_n\}|}\sum_{\{k\;:\;\mu(k,W)\geq \varepsilon_n\}} \langle e_k,Te_k\rangle
\end{equation}
is an equivalent formulation for any decreasing, strictly positive sequence $\varepsilon_n \to 0$ such that the sets $\{k\;:\;\mu(k,W)\geq \varepsilon_n\}$ satisfy a certain growth condition.

\begin{lemma}\label{L: subsequence summation}
Let $\{a_k\}_{k \in \mathbb{N}} \subseteq \mathbb{R}$ be a bounded sequence and let $\{k_1, k_2, \dots \}$ be an infinite, increasing sequence of positive integers such that \[\lim_{n\to \infty} \frac{k_{n+1}}{k_n} = 1.\] Then $\lim_{i \to \infty} \frac{1}{1+k_i} \sum_{k=1}^{k_i} a_k$ exists if and only if $\lim_{n\to \infty} \frac{1}{1+n} \sum_{k=1}^n a_k$ exists. 
\begin{proof}
Put $b_n := \frac{1}{1+n} \sum_{k=1}^n a_k$ and note that trivially $\{b_{k_i}\}_{i \in \mathbb{N}}$ converges if $\{b_n\}_{n\in \mathbb{N}}$ does.

To prove the converse, suppose that $\{b_{k_i}\}$ converges to limit $L$. Let $M \in \mathbb{R}$ be such that $\abs{a_k} \leq M$ for all $k$, then by shifting to the sequence $\{a_k + M\}_{k \in \mathbb{N}}$ we can assume without loss of generality that $\{a_k\}_{k\in \mathbb{N}}$ is a positive sequence. 

Define for each $n$ the integer $k_{i_n} := \max\{k_{i} \leq n : i\in \mathbb{N} \}.$
We will now prove that \[|b_n - b_{k_{i_n}}| \xrightarrow{n \to \infty} 0,\] as this would mean that \[\abs{|b_n - L| - |L - b_{k_{i_n}}|} \leq |b_n - b_{k_{i_n}}|\xrightarrow{n \to \infty} 0,\] which forces that $\{b_n\}$ converges to $L$ as well.

Observe that \begin{align*}
    \frac{1}{1+n} \sum_{k=1}^n a_k -  \frac{1}{1+k_{i_n}} \sum_{k=1}^{k_{i_n}} a_k 
    & \leq \frac{1}{1+k_{i_n}} \sum_{k=1}^{k_{i_n+1}} a_k -  \frac{1}{1+k_{i_n}} \sum_{k=1}^{k_{i_n}} a_k\\
    &= \frac{1+k_{i_n+1}}{1+k_{i_n}} b_{k_{i_n+1}} - b_{k_{i_n}}\\
    &\xrightarrow{n\to\infty} 0;\\
    \frac{1}{1+k_{i_n}} \sum_{k=1}^{k_{i_n}} a_k-\frac{1}{1+n} \sum_{k=1}^n a_k & \leq \frac{1}{1+k_{i_n}} \sum_{k=1}^{k_{i_n}} a_k-\frac{1}{1+k_{i_n+1}} \sum_{k=1}^{k_{i_n}} a_k\\
    &= \left(1 - \frac{1+k_{i_n}}{1+k_{i_n+1}} \right) b_{k_{i_n}}\\
    & \xrightarrow{n\to \infty} 0.\qedhere
\end{align*}
\end{proof}
\end{lemma}

This lemma, combined with \eqref{feq} immediately implies the following proposition.

\begin{proposition}\label{general_DOS_proposition}
    If $T$ is a bounded linear operator and $0 < W \in \Lc_{1,\infty}$ such that \[\lim_{n\to\infty} \frac{|\{ k\geq 0 : \mu(k,W) \geq \varepsilon_{n+1} \}|}{|\{ k\geq 0 : \mu(k,W) \geq \varepsilon_n \}|} = 1.\] for some decreasing, strictly positive sequence $\varepsilon_n \to 0$, then for all extended limits $\omega$ we have 
    \begin{equation}\label{abstract_DOS}
        \tr_\omega(TW) = \tr_\omega(W)\lim_{n\to \infty} \frac{\tr(T\chi_{[\varepsilon_n,\infty)}(W))}{\tr(\chi_{[\varepsilon_n,\infty)}(W))},
    \end{equation}
    whenever the limit on the right hand side exists.
\end{proposition}
This proposition gives us the density of states formula for the discrete case. The idea is that in $\ell_2(X)$ we consider the basis $\{\delta_v\}_{v \in X}$ of indicator functions
of points $p \in X$, and $W$ is an operator of pointwise multiplication by a function which is \emph{radially decreasing} with respect to some point $x_0 \in X$, so that $\chi_{[\varepsilon_n,\infty)}(W)$ is the indicator function of a ball, and the limit $n\to \infty$ is equivalent to taking a limit over balls with radius going to infinity. This is where property~\eqref{E: Condition} as discussed in Section~\ref{S: Metric Condition} comes in, as this ensures that the premise of Proposition \ref{general_DOS_proposition} is satisfied.

{\it Proof of Theorem~\ref{T: Main}.} 
    Let $W = M_w$. By the assumptions of the theorem, $W \geq 0$. If $W$ is trace-class, the theorem is trivial as both sides are zero. So assume that $0 < W \in \Lc_{1,\infty}(\hilb) \setminus \Lc_1(\hilb)$. Then $\{\delta_v\}_{v \in X}$ is a basis of normalised eigenvectors for $W$, with eigenvalue corresponding to $\delta_v$ equal to $w(v)$. By assumption $w(v)$
    is a strictly decreasing function of $d_X(x_0,v)$, and hence the sets
       $\{v\; \colon \;w(v)\geq \delta\}$
    are balls. In fact, if we define $\varepsilon_n = \tilde{\mu}(n, W)$, where $\tilde{\mu}(n, W)$ is the $n$th largest singular value of $W$ counted \textit{without} multiplicities, we have that \[|\{ k\geq 0 : \mu(k,W) \geq \varepsilon_n \}|=|\{v\; \colon \;w(v)\geq \varepsilon_n\}| = |B(x_0, r_n)|,\] and we have assumed that \[\lim_{n\rightarrow \infty}\frac{|B(x_0, r_{k+1})|}{|B(x_0, r_k)|} = 1.\] Hence due to Lemma~\ref{expectation_values_lemma} and Lemma~\ref{L: subsequence summation}, we can then conclude that
    \begin{align*}
        \tr_\omega(TW) &= \tr_\omega(W)\lim_{k\to \infty} \frac{1}{|B(x_0,r_k)|}\sum_{v \in B(x_0,r_k)} \langle \delta_v,T\delta_v\rangle
    \end{align*}
    since we have assumed that the limit on the right exists.
    Observing that the sum on the right side is equal to $\tr(TM_{\chi_{B(0,r_k)}})$
    gives \eqref{E: general DOS formula}. \hspace*{\fill} $\Box$
   
\medskip

\begin{remark}
Note that the direct cause of why we put condition~\eqref{E: Condition} on $X$ is because we wanted to make sure that \[\lim_{n\to\infty} \frac{|\{ k\geq 0 : \mu(k,M_w) \geq \varepsilon_{n+1} \}|}{|\{ k\geq 0 : \mu(k,M_w) \geq \varepsilon_n \}|} = 1.\] If we again denote $\tilde{\mu}(n, M_w)$ as the $n$th largest singular value of $M_w$ counted \textit{without} multiplicities, define $m_n$ as the multiplicity corresponding to this singular value and also define $M_n := \sum_{k=1}^n m_k$, then we have effectively imposed \[\lim_{n\rightarrow \infty} \frac{M_{n+1}}{M_n}.\] Now compare this to the preprint by Cipriani and Sauvageot~\cite{CiprianiSauvageot2021} also referenced in Section~\ref{S: Metric Condition}. They study densely defined, nonnegative, unbounded, self-adjoint operators with exactly such a property for the multiplicities of its eigenvalues, and our $(M_w)^{-1}$ would fit such a description. The link between the DOS as considered in this paper and the spectral weight those authors define remains unclear as of yet.
\end{remark}

\appendix

\section{Problems of measurability}
\label{S: measurability}
Going back to the Toeplitz lemma~\ref{toeplitz_lemma}, it is possible to get better behaviour of the convergence $\frac{1}{b_n}\sum_{k=0}^n a_kx_k\to L$ 
by assuming faster convergence of $x_k\to L$. For example, we have the following lemma with an obvious proof. 
\begin{lemma} \label{L: unused lemma}
    Let $a_n$ and $b_n$ satisfy the same assumptions as Lemma~\ref{toeplitz_lemma}. 
    Let $x_n \in \mbC, \ n=0,1,\ldots.$ 
    If $x_k\rightarrow L$ sufficiently fast such that $\{a_k|x_k - L|\}_{k=0}^\infty \in \ell_1$, then:
    \begin{equation*}
        \sum_{k=0}^n a_kx_k = L b_n + O(1).
    \end{equation*}
\end{lemma}

\begin{lemma} \label{second unused lemma}
    Let $a_n$ and $b_n$ satisfy the same assumptions as Lemma~\ref{toeplitz_corollary}. 
    Let $x_n \in \mbC, \ n=0,1,\ldots$ and $\sigma_n = \frac 1{n+1} \sum_{k=0}^n x_k.$
    If $\sigma_n \rightarrow L$ sufficiently fast such that $\{a_k| \sigma_k - L|\}_{k=0}^\infty \in \ell_1$, then:
    \begin{equation*}
        \sum_{k=0}^n a_kx_k = L b_n + O(1).
    \end{equation*}
\end{lemma}
\begin{proof}
    Let $y_n = \sum_{k=0}^{n-1} x_k$ with $y_{0}=0$. From the proof of Lemma~\ref{toeplitz_corollary} we have: 
    \begin{align*}
        \sum_{k=0}^n a_kx_k 
                            = y_{n+1}a_n+\sum_{k=1}^n \sigma_k \cdot k(a_{k-1}-a_{k}),
    \end{align*}
    the sequence $a_k' := k(a_{k-1}-a_{k})$ is non-negative, and as $n\to\infty$,
    \begin{equation*}
        \sum_{k=1}^n a_k' = b_{n-1}-na_n\rightarrow \infty,
    \end{equation*}
    Therefore, since by assumption $\lim_{k\to\infty} \sigma_k = L$
and $\{a_k| \sigma_k - L|\}_{k=0}^\infty \in \ell_1,$
     Lemma~\ref{L: unused lemma} applies with $a'$ in place of $a$ and $\sigma_k$ in place of $x_k,$ which gives 
    \begin{equation*}
        \sum_{k=1}^n \sigma_k\cdot k(a_{k-1}-a_k) = L (b_{n-1}-na_n) + O(1)
    \end{equation*} 
    as $n\to\infty$.
    Thus,
    \begin{align*}
        \sum_{k=0}^n a_kx_k  & = \frac{y_{n+1}}{n}\cdot na_n + L(b_{n-1}-na_n) + O(1)  \\
           & = L b_{n-1} + O(1) = L b_{n} + O(1).
    \end{align*}
\end{proof}

\medskip
For an operator $A \in \Lc_{1,\infty},$ there are various criteria relating the behaviour of the sequence $\sum_{k=0}^n \lambda(k,A)$ to the measurability of $A$. For example,~\cite[Theorem 5.1.5]{LSZVol1} implies that
\begin{equation} \label{F: what a beauty}
    \sum_{k=0}^n \lambda(k,A) - c\log(2+n)=O(1),\quad n\to \infty
\end{equation}
if and only if $\varphi(A) = c$ for all normalised traces $\varphi$ on $\Lc_{1,\infty}$ (c.f.~\cite[Theorem 9.1.2]{LSZVol1}). For different classes of traces, different criteria are available, see~\cite{SemenovSukochev2015}.

\begin{theorem}
    Let $(X,d_X)$, $T$ and $w$ satisfy the assumptions of Theorem~\ref{T: Main}.
    Let~$e_n$ be an orthonormal basis such that $M_w e_k = w^*(k) e_k.$
     If $$\frac{1}{n+1}\sum_{k=0}^n \langle e_k,Te_k\rangle \rightarrow L \in \Cplx$$
    so fast that
    \begin{equation*}
        \sum_{n=0}^\infty w^*(n)\left| \frac{1}{n+1}\sum_{k=0}^n \langle e_k,Te_k\rangle - L\right| < \infty
    \end{equation*}
    and there exists $C > 0$ such that:
    \begin{equation*}
        \sum_{k=0}^n w^*(k) = C\log(2+n)+O(1),
    \end{equation*}
    then $TM_w$ is measurable in the sense of Connes, specifically
    $$
       \varphi(TM_w) = \tr_\omega(M_w) \lim_{n \to \infty} \frac 1{n+1} \sum_{k=0}^n \langle e_k,Te_k\rangle
    $$
    for all traces $\varphi$ on $\Lc_{1,\infty}.$
\end{theorem}
\begin{proof} 
Since $TM_w$ is $M_w$-modulated (see Definition~\ref{D: V-modulated}), from \eqref{expectation_values_sum} we have
$$
   \sum_{k=0}^n \lambda(k,TM_w) = \sum_{k=0}^n \Scal{e_k,Te_k} w^*(k) + O(1).
$$
Hence, in view of \eqref{F: what a beauty} to prove the claim it suffices to show that 
$$
   \sum_{k=0}^n \Scal{e_k,Te_k} w^*(k) = CL \log(2+n) + O(1).
$$
By the second condition this is equivalent to 
$$
   \sum_{k=0}^n \Scal{e_k,Te_k} w^*(k) = L \sum_{k=0}^n w^*(k) + O(1),
$$
so it suffices to prove this. This follows from Lemma~\ref{second unused lemma} applied to $x_k = \Scal{e_k,Te_k}$ and $a_k = w^*(k).$ 
\end{proof}

We could also replace the assumption with the slightly stronger assertion:
\begin{equation*}
    \left\{\frac{1}{n+1}\sum_{k=0}^n \langle e_k,Te_k\rangle-L\right\}_{k=0}^\infty \in \Lambda_{\log},
\end{equation*}
where $\Lambda_{\log}$ is the space of sequences $x$ such that
\begin{equation*}
    \sum_{k=0}^\infty \frac{x^*_k}{k+1} < \infty.
\end{equation*}

\section{Equivariance of the DOS under translations of the Hamiltonian}
\label{S: Translation}
In this appendix we will provide a straightforward application of the Dixmier formula for the DOS put forward in Theorem~\ref{T: Main}. Namely, we provide a new and original proof of the equivariance of the DOS on lattice graphs $X$ under translations of the Hamiltionian. By this we mean that if $U$ denotes a shift operator on $\ell_2(X)$ and the DOS exists for both a Hamiltonian $H$ and the shifted $UHU^*$, then the DOS is equal for $H$ and $UHU^*$. This fact is not hard to prove without Theorem~\ref{T: Main}, but it does provide a different perspective on the claim.

Afterwards, we will discuss some consequences of this translation equivariance.
    
\subsection{Translation equivariance on lattice graphs}
    We will consider the example where $X = \Itgr^d$, embedded as a subset of $\Rl^d$ with the Euclidean metric. Precisely the same reasoning
    applies to other discrete subsets $X\subset \Rl^d$, such as lattices (recall that a lattice in $\Rl^d$ is the $\Itgr$-linear span of $d$ linearly independent
    vectors).
    
    We take 
    \begin{equation} \label{F: w(x) = (x)}
       w(x) := (1+\norm{x}_2)^{-d}
    \end{equation}
    where $x \in \Itgr^d$ like before.
    
    \begin{lemma}\label{translation_difference_L1}
        For all $n\in \Itgr^d$, we have:
        \begin{equation*}
            \{w(x)-w(x-n)\}_{x \in \Itgr^d} \in \ell_{\frac{d}{d+1},\infty}.
        \end{equation*}
    \end{lemma}
    \begin{proof}
        The difference $w(x)-w(x-n)$ is provided by the formula:
        \begin{equation*}
            w(x)-w(x-n) = \int_0^1 \langle \nabla w(x-(1-\theta)n),n\rangle\,d\theta.
        \end{equation*}
        The gradient $\nabla w$ of $w$ is easily computed as:
        \begin{equation*}
            \frac{\partial}{\partial x_j}w(x) = -\frac{dx_j}{\norm{x}_2 (1+\norm{x}_2)^{d+1}}.
        \end{equation*}
        Thus,
        \begin{equation*}
            \norm{\nabla w(x)}_2 = \frac{d}{(1+\norm{x}_2)^{d+1}}.
        \end{equation*}
        Therefore for $x$ with $\norm{x}_2 > \norm{n}$ we have
        \begin{align*}
            \abs{w(x)-w(x-n)} &\leq d\norm{n}_2 \max_{0\leq \theta\leq 1}(1+\norm{x-(1-\theta) n)}_2)^{-d-1} \\
            &\leq  \frac{d \norm{n}_2}{(1+\norm{x}_2 - \norm{n}_2)^{d+1}},
        \end{align*}
       hence $\abs{w(x) - w(x-n)}_{x\in \mathbb{Z}^d} $ is an element of $\ell_{\frac{d}{d+1}, \infty}$.
    \end{proof} 
    
    \begin{theorem}
        For $n \in \Itgr^d$, let $U_n$ denote the operator on $\ell_2(\Itgr^d)$ of translation by~$n$. Assume that $H = H_0+M_V$ is a Hamiltonian operator
        such that the density of states exists for both $H$ and $U_nHU_n^*.$ Then both measures are equal.
    \end{theorem}
    \begin{proof}
        For any $f \in C_c(\Rl)$, we have
            $f(U_nHU_n^*) = U_nf(H)U_n^*.$ Combining this with 
 the tracial property of the Dixmier trace gives 
        \begin{equation*}
            \tr_\omega(f(U_nHU_n^*)M_w) = \tr_\omega(f(H)U_n^*M_wU_n).
        \end{equation*}
        By Lemma~\ref{translation_difference_L1}, we have
            $U_n^*M_wU_n-M_w \in \Lc_{\frac{d}{d+1},\infty}\subset \Lc_1.$
Since the Dixmier trace vanishes on the trace class, it follows that 
            $\tr_\omega(f(U_nHU_n^*)M_w) = \tr_\omega(f(H)M_w).$
 Combining this with the Dixmier trace formula for the density of states, \eqref{F: DOS integral}, completes the proof. 
    \end{proof}
    
    Via identical reasoning, we also have the following abstract assertion:
    \begin{theorem}\label{general_translation_invariance}
        Let $(X,d_X)$ be an infinite discrete metric space and $w$ be a function such that these satisfy the assumptions of Theorem~\ref{T: Main}, let $\gamma$ be an isometry of $X$, and let $U_{\gamma}\delta_p = \delta_{\gamma(p)}$ be the corresponding unitary operator on $\ell_2(X)$. Assume that
            $w-w\circ\gamma \in (\ell_{1,\infty})_0(X).$
        Then 
        \begin{equation*}
            \lim_{R\to \infty} \frac{1}{|B(x_0,R)|}\sum_{d_X(p,x_0)\leq R} \langle \delta_p,T\delta_p\rangle = \lim_{R\to\infty} \frac{1}{B(x_0,R)}\sum_{d_X(p,x_0)\leq R} \langle \delta_p,U_{\gamma}^*TU_{\gamma}\rangle,
        \end{equation*}        
        provided \emph{both} limits exist.
    \end{theorem}
    
\subsection{Ergodic operators}
As emphasised in the introduction of this appendix, the following results are direct consequences of the translation equivariance of the DOS measure and therefore could be derived without help of Theorem~\ref{T: Main}.

Let $(\Omega,\Sigma,\mathbb{P})$ be a probability space, and let $\Gamma$ be a discrete amenable group of isometries of the metric space $X$
 from Theorem~\ref{T: Main}. We assume that there
    is a representation of $\Gamma$ as automorphisms of $\Omega$:
    \begin{equation*}
        \gamma\in \Gamma\mapsto \alpha_\gamma\in \mathrm{Aut}(\Omega).
    \end{equation*}
    It is assumed that the action $\alpha$ is ergodic, in the sense that:
    \begin{enumerate}[{\rm (i)}]
        \item{} For every $\gamma\in \Gamma$, the automorphism $\alpha_\gamma$ is measure preserving;
        \item{} If $E\subseteq \Omega$ is invariant under every $\alpha_\gamma$, then $\mathbb{P}(E)=0$ or $\mathbb{P}(\Omega\setminus E) = 0$.
    \end{enumerate}
    
    We will use a generalisation of Birkhoff's ergodic theorem, obtained by Lindenstrauss~\cite[Theorem 1.3]{Lindenstrauss2001}. This uses the concept of a F\o lner sequence, we give the definition as it is used for discrete groups.

\begin{definition}
Let $\Gamma$ be a discrete group, and let $\{F_n\}_{n=0}^\infty$ be a sequence of subsets of $\Gamma$. 
    \begin{enumerate}[{\rm (i)}]
        \item{} If for every finite subset $K \subseteq \Gamma$ and every $\delta>0$, there exists $N$ sufficiently large such that if $n>N$, we have for all $k\in K$
        $$|F_n \ \Delta \ kF_n| \leq \delta|F_n|,$$ then $\{F_n\}_{n=0}^\infty$ is called a F\o lner sequence.
        \item{} If $\{F_n\}_{n=0}^\infty$ satisfies (i) and furthermore for some $C \geq 1$ and for every $n\geq 0$, we have:
        $$\left|\bigcup_{k\leq n} F_k^{-1}F_{n+1}\right| \leq C|F_{n+1}|,$$ then $\{F_n\}_{n=0}^\infty$ is called a tempered F\o lner sequence. 
    \end{enumerate}
\end{definition}

The existence of a F\o lner sequence in this sense is equivalent with the condition of $\Gamma$ being discrete and amenable~\cite[p.~23]{Lubotzky1994}. Also note that any F\o lner sequence has a tempered subsequence~\cite[Proposition~1.4]{Lindenstrauss2001}.

    Lindenstrauss' pointwise ergodic theorem~\cite[Theorem 1.3]{Lindenstrauss2001} implies that if $\{F_n\}_{n=0}^\infty$
    is a tempered F\o lner sequence, then for all $f \in L_1(\Omega)$ we have:
    \begin{equation} \label{Lindenstrauss' Thm}
        \lim_{n\to\infty} \frac{1}{|F_n|} \sum_{\gamma\in F_n} f(\alpha_\gamma\omega) = \mathbb{E}(f).
    \end{equation}
    for almost every $\omega \in \Omega.$ 
    
    For $\gamma\in \Gamma$, let $U_{\gamma}$ denote the induced unitary operator acting on $\ell_2(X)$ by:
    \begin{equation*}
        U_\gamma\delta_p := \delta_{\gamma(p)},\quad p\in X,\,\gamma\in \Gamma.
    \end{equation*}
    
    We will consider strongly measurable random operators $T \in L_1(\Omega,\Bc(\ell_2(X))$ which are compatible with $\alpha$ in the sense that:
    \begin{equation}\label{group_compatibility}
        U_{\gamma}T(\omega)U_\gamma^* = T(\alpha_\gamma \omega),\quad \gamma\in \Gamma
    \end{equation}
    for almost all $\omega \in \Omega.$ 
        
    \begin{proposition}
        Let $T \in L_1(\Omega,\Bc(\ell_2(X))$ be a random operator satisfying \eqref{group_compatibility} with respect to a group of isometries $\Gamma$ of $X,$
        which admits a tempered F\o lner sequence $\set{F_n}_{n=0}^\infty$ of finite subsets,  
         and with respect to an ergodic action $\alpha$
        of $\Gamma$ on $\Omega$. If there exists a function $w \colon X \to \mbR_+$ satisfying the assumptions of Theorem~\ref{T: Main} such that
        \begin{equation*}
            w\circ \gamma-w \in (\ell_{1,\infty})_0(X)
        \end{equation*} 
        for every $\gamma\in \Gamma$ then the density of states of $T(\xi)$ is non-random, in the sense that if the limit:
        \begin{equation*}
            \lim_{R\to\infty} \frac{\tr(T(\xi) M_{\chi_{B(x_0,R)}})}{|B(x_0,R)|}
        \end{equation*}
        exists for almost every $\xi$, then the limit is almost surely constant in $\xi$.        
    \end{proposition}
    \begin{proof}
        This is an application of the Lindenstrauss' version of Birkhoff's ergodic theorem. 
        The assumption on $w$ and Theorem~\ref{general_translation_invariance} imply that:
        \begin{equation} \label{alpha invariant}
            \tr_\omega(T(\xi) M_w) = \tr_\omega(T(\alpha_\gamma\xi)M_w),\quad \gamma\in \Gamma.
        \end{equation}  
        Therefore for every $n\geq 0$ we have:
        \begin{equation*}
            \tr_\omega(T(\xi) M_w) = \frac{1}{|F_n|}\sum_{\gamma\in F_n} \tr_\omega(T(\alpha_\gamma\xi)M_w).
        \end{equation*}
        Note that:
        \begin{equation*}
            |\tr_\omega(T(\xi)M_w)| \leq \|T(\xi)\|_\infty\|w\|_{1,\infty}.
        \end{equation*}
        Hence the function $\xi\mapsto \tr_\omega(T(\xi)M_w)$ is integrable, due to our assumption that $T\in L_1(\Omega,\Bc(\ell_2(X)))$, and the measurability
        of $\xi\mapsto \tr_\omega(T(\xi)M_w)$ follows from the strong measurability of $\xi\mapsto T(\xi)$ and the norm continuity of $T\mapsto \tr_\omega(TM_w)$.
        Hence, Lindenstrauss' ergodic theorem \eqref{Lindenstrauss' Thm}
        applies to this function, and hence for almost every $\xi\in \Omega$ we have:
        \begin{equation*}
            \lim_{n\to\infty} \frac{1}{|F_n|} \sum_{\gamma \in F_n} \tr_\omega(T(\alpha_\gamma\xi)M_w) = \mathbb{E}(\tr_\omega(TM_w)).
        \end{equation*}
        The right hand side has no dependence on $\xi\in \Omega$, and hence the limit is almost surely constant in $\xi$.
  Due to \eqref{alpha invariant}, this implies that $\tr_\omega(T(\xi) M_w)$ is almost surely constant in $\xi.$ Alluding to Theorem~\ref{T: Main}, 
  we conclude that the density of states of $T(\xi)$ is almost surely constant in $\xi.$
    \end{proof}
    
    In an alternative direction of inquiry, the condition \eqref{group_compatibility} can be used in some circumstances to imply the existence of the density of states.
    For simplicity, we state the following condition for $X = \Itgr^d$.
    \begin{theorem} \label{T: on existence of DOS}
        Let $T\in L_1(\Omega,\Bc(\ell_2(\Itgr^d)))$ be a linear operator which satisfies \eqref{group_compatibility} with respect to the action
        of $\Itgr^d$ on itself by translations and an ergodic action $\alpha$ of~$\mbZ^d$ on $\Omega.$ Then for almost every $\xi \in \Omega$ there exists the limit:
        \begin{equation*}
            \lim_{R\to\infty} \frac{1}{|B(0,R)|}\sum_{|n|\leq R} \langle \delta_n,T(\xi)\delta_n\rangle = \mathbb{E}(\langle \delta_0,T\delta_0\rangle).
        \end{equation*}
    \end{theorem}
    \begin{proof}
        We have that $U_n\delta_0 = \delta_n$, and therefore:
        \begin{equation*}
            \langle \delta_n,T(\xi)\delta_n\rangle = \langle \delta_0,U_n^*T(\xi)U_n\delta_0\rangle = \langle \delta_0,T(\alpha_{-n}\xi)\delta_0\rangle.
        \end{equation*}
        It follows that:
        \begin{equation*}
            \frac{1}{|B(0,R)|}\sum_{|n|\leq R} \langle \delta_n,T(\xi)\delta_n\rangle = \frac{1}{|B(0,R)|}\sum_{|n|\leq R}\langle \delta_0,T(\alpha_n\xi)\delta_0\rangle.
        \end{equation*}
        By our assumption on $T$, the function $\xi \mapsto \langle \delta_0,T(\xi)\delta_0\rangle$ belongs to $L_1(\Omega)$.
        Note that the sequence $F_N := B(0,N)$ is a tempered F\o lner sequence in $\Itgr^d$, and hence Lindenstrauss' ergodic theorem \eqref{Lindenstrauss' Thm}  implies that
        for almost every $\xi\in \Omega$ there exists the limit
        \begin{equation*}
            \lim_{N\to\infty} \frac{1}{|B(0,N)|} \sum_{n\in B(0,N)} \langle \delta_0,T(\alpha_n \xi)\delta_0\rangle = \mathbb{E}(\langle \delta_0,T(\xi)\delta_0\rangle).\qedhere
        \end{equation*}
    \end{proof}
    Note that the result also holds if the limit over balls $\{B(0,N)\}_{N\geq 0}$ is replaced with any other tempered F\o lner sequence, such as cubes $\{[-N,N]^d\}_{N\geq 0}$.
    The limit in every case is $\mathbb{E}(\langle \delta_0,T(\xi)\delta_0\rangle)$, and hence does not depend on the choice of sequence of sets.

    \begin{theorem}
      Let $H(\xi) = H_0 + V_\xi(x)$ be a random operator on $\ell_2(\mbZ^d),$
       where $H_0$ is a $\mbZ^d$-translation invariant difference operator and $V_\xi,$ $\xi \in \Omega,$ 
      an iid random bounded function. Then there exists a set $\Omega_0 \subset \Omega$ of probability $1,$
      such that for any $f \in C_c(\mbR)$ and for any $\xi \in \Omega_0$ 
 there exists the limit:
        \begin{equation} \label{F: what we need}
            \lim_{R\to\infty} \frac{1}{|B(0,R)|}\sum_{|n|\leq R} \langle \delta_n, f(H_\xi)\delta_n\rangle = \mathbb{E}(\langle \delta_0,f(H_\xi)\delta_0\rangle).
        \end{equation}
    \end{theorem}
\begin{proof} Proof follows a standard argument, see e.g.~\cite[Chapter 3]{AizenmanWarzel2015}.
Let $\Sigma$ be a countable dense subset  of $C_c(\mbR).$ 
The random operator $H(\xi)$ is ergodic and obeys \eqref{group_compatibility} so Theorem~\ref{T: on existence of DOS} is applicable. 
By this theorem, for every $f \in \Sigma$ 
there exists a full set $\Omega_f \subset \Omega$ such that \eqref{F: what we need} holds for all $\xi \in \Omega_f.$ Define a full set $\Omega_0 = \bigcap_{f \in \Sigma} \Omega_f,$
so for every $f \in \Sigma$ and every $\xi \in \Omega_0$ the equality \eqref{F: what we need} holds. Choose any $g \in C_c(\mbR)$ and let $f_1,f_2, \ldots \in \Sigma$ be such that 
$f_n \to g$ in uniform topology. Let $\varepsilon>0.$ Further we proceed by a standard $\varepsilon/3$-trick. Let $N \in \mbN$ be such that for all $n\geq N$ 
$\norm{f_n-g}_\infty < \varepsilon/3.$ For $f_N$ the equality \eqref{F: what we need} holds for any $\xi \in \Omega _0.$ Let $R_0 > 0$ be such that for all $R > R_0$ 
and all $\xi \in \Omega _0$
$$
    \abs{ \frac{1}{|B(0,R)|}\sum_{|n|\leq R} \langle \delta_n, f_N(H_\xi)\delta_n\rangle - \mathbb{E}(\langle \delta_0,f_N(H_\xi)\delta_0\rangle) }  < \varepsilon/3.
$$
Then for all $R> R_0$ and $\xi \in \Omega_0$ we have 
\begin{align*}
    \Big| \frac{1}{|B(0,R)|}  &  \sum_{|n|\leq R} \langle \delta_n, g(H_\xi)\delta_n\rangle - \mathbb{E}(\langle \delta_0, g(H_\xi)\delta_0\rangle) \Big|  \\
      & \leq  \abs{ \frac{1}{|B(0,R)|}\sum_{|n|\leq R} \langle \delta_n, [g(H_\xi) - f_N(\xi)]   \delta_n\rangle  } \\ 
     & \qquad  +  \Big| \frac{1}{|B(0,R)|}  \sum_{|n|\leq R} \langle \delta_n, f_N(H_\xi)\delta_n\rangle - \mathbb{E}(\langle \delta_0, f_N(H_\xi)\delta_0\rangle) \Big|  \\
     & \qquad  + \Big|   \mathbb{E}(\langle \delta_0, [ f_N(H_\xi) - g(H_\xi) ]  \delta_0\rangle)   \Big|  \\
     & < \varepsilon,
\end{align*}
where the last inequality follows from the triangle, Schwartz and $\norm{f(H) - g(H)} \leq \norm{f-g}_\infty < \varepsilon/3$ inequalities. 
\end{proof}


\begin{thebibliography}{10}

\bibitem{AdachiSunada1993}
T.~Adachi and T.~Sunada.
\newblock Density of states in spectral geometry.
\newblock {\em Commentarii Mathematici Helvetici}, 68(1):480--493, 1993.

\bibitem{AizenmanWarzel2015}
M.~Aizenman and S.~Warzel.
\newblock {\em Random operators: Disorder effects on Quantum Spectra and
  Dynamics}, volume 168 of {\em Graduate Studies in Mathematics}.
\newblock American Mathematical Soc., 2015.

\bibitem{AlexanderOrbach1982}
S.~Alexander and R.~Orbach.
\newblock Density of states on fractals:``fractons''.
\newblock {\em Journal de Physique Lettres}, 43(17):625--631, 1982.

\bibitem{Ashcroft1976}
N.~W. Ashcroft and N.~D. Mermin.
\newblock {\em Solid state physics}, volume 120.
\newblock Saunders, 1976.

\bibitem{AMSZ}
N.~Azamov, E.~McDonald, F.~Sukochev, and D.~Zanin.
\newblock A {D}ixmier trace formula for the density of states.
\newblock {\em Communications in Mathematical Physics}, 377(3):2597--2628,
  2020.

\bibitem{AMZS_domains}
N.~Azamov, E.~McDonald, D.~Zanin, and F.~Sukochev.
\newblock The density of states depends on the domain.
\newblock {\em arXiv preprint arXiv:2107.09828}, 2021.

\bibitem{BaakeGrimm2006}
M.~Baake and U.~Grimm.
\newblock Averaged coordination numbers of planar aperiodic tilings.
\newblock {\em Philosophical Magazine}, 86(3-5):567--572, 2006.

\bibitem{BerezinShubin1991}
F.~A. Berezin and M.~Shubin.
\newblock {\em The {S}chr{\"o}dinger Equation}, volume~66 of {\em Mathematics
  and its applications, Soviet series}.
\newblock Kluwer Academic Publishers, 1991.

\bibitem{Bloch1929}
F.~Bloch.
\newblock {Z}ur {S}uszeptibilit{\"a}t und {W}iderstands{\"a}nderung der
  {M}etalle im {M}agnetfeld.
\newblock {\em Zeitschrift f{\"u}r Physik}, 53(3):216--227, 1929.

\bibitem{BourgainKlein2013}
J.~Bourgain and A.~Klein.
\newblock Bounds on the density of states for {S}chr{\"o}dinger operators.
\newblock {\em Inventiones mathematicae}, 194(1):41--72, 2013.

\bibitem{Brunner1979}
G.~O. Brunner.
\newblock The properties of coordination sequences and conclusions regarding
  the lowest possible density of zeolites.
\newblock {\em Journal of Solid State Chemistry}, 29(1):41--45, 1979.

\bibitem{CarmonaLacroix1990}
R.~Carmona and J.~Lacroix.
\newblock {\em Spectral theory of random Schr{\"o}dinger operators}.
\newblock Birkh{\"a}user, 1990.

\bibitem{CerfTheret2016}
R.~Cerf and M.~Th{\'e}ret.
\newblock Weak shape theorem in first passage percolation with infinite passage
  times.
\newblock {\em Annales de l'Institut Henri Poincar{\'e}, Probabilit{\'e}s et
  Statistiques}, 52(3):1351--1381, 2016.

\bibitem{ChayesChayes1986}
J.~T. Chayes, L.~Chayes, J.~R. Franz, J.~P. Sethna, and S.~A. Trugman.
\newblock On the density of state for the quantum percolation problem.
\newblock {\em Journal of Physics A: Mathematical and General},
  19(18):L1173--L1177, 1986.

\bibitem{CiprianiSauvageot2021}
F.~E.~G. Cipriani and J.-L. Sauvageot.
\newblock Measurability, spectral densities and hypertraces in noncommutative
  geometry.
\newblock {\em arXiv preprint arXiv:2111.15575}, 2021.

\bibitem{ConwaySloane1993}
J.~H. Conway and N.~J.~A. Sloane.
\newblock {\em Sphere packings, lattices and groups}.
\newblock Springer Science \& Business Media, 1993.

\bibitem{ConwaySloane1997}
J.~H. Conway and N.~J.~A. Sloane.
\newblock Low--dimensional lattices. {VII} coordination sequences.
\newblock {\em Proceedings of the Royal Society of London. Series A:
  Mathematical, Physical and Engineering Sciences}, 453(1966):2369--2389, 1997.

\bibitem{deOliveiraPrado2005}
C.~R. de~Oliveira and R.~A. Prado.
\newblock Dynamical delocalization for the 1{D} {B}ernoulli discrete {D}irac
  operator.
\newblock {\em Journal of Physics A: Mathematical and General},
  38(7):L115--L119, 2005.

\bibitem{DoiIwatsuka2001}
S.-i. Doi, A.~Iwatsuka, and T.~Mine.
\newblock The uniqueness of the integrated density of states for the
  {S}chr{\"o}dinger operators with magnetic fields.
\newblock {\em Mathematische Zeitschrift}, 237(2):335--371, 2001.

\bibitem{EnglischKirsch1988}
H.~Englisch, W.~Kirsch, M.~Schroder, and B.~Simon.
\newblock Density of surface states in discrete models.
\newblock {\em Physical Review Letters}, 61(11):1261, 1988.

\bibitem{Grimmett1993}
G.~R. Grimmett.
\newblock {\em Percolation}.
\newblock Springer, 1999.

\bibitem{HislopMuller2008}
P.~Hislop and P.~M{\"u}ller.
\newblock A lower bound for the density of states of the lattice {A}nderson
  model.
\newblock {\em Proceedings of the American Mathematical Society},
  136(8):2887--2893, 2008.

\bibitem{HoddesonBaym1987}
L.~Hoddeson, G.~Baym, and M.~Eckert.
\newblock The development of the quantum-mechanical electron theory of metals:
  1928---1933.
\newblock {\em Rev. Mod. Phys.}, 59:287--327, Jan 1987.

\bibitem{Hof1993}
A.~Hof.
\newblock Some remarks on discrete aperiodic {S}chr{\"o}dinger operators.
\newblock {\em Journal of statistical physics}, 72(5):1353--1374, 1993.

\bibitem{Hof1995}
A.~Hof.
\newblock A remark on {S}chr{\"o}dinger operators on aperiodic tilings.
\newblock {\em Journal of statistical physics}, 81(3):851--855, 1995.

\bibitem{Jaric1989}
M.~V. Jaric, editor.
\newblock {\em Introduction to the Mathematics of Quasicrystals}.
\newblock Academic Press, 1989.

\bibitem{KaltonLord2013}
N.~Kalton, S.~Lord, D.~Potapov, and F.~Sukochev.
\newblock Traces of compact operators and the noncommutative residue.
\newblock {\em Advances in Mathematics}, 235:1--55, 2013.

\bibitem{KirschMuller2006}
W.~Kirsch and P.~M{\"u}ller.
\newblock Spectral properties of the {L}aplacian on bond-percolation graphs.
\newblock {\em Mathematische Zeitschrift}, 252(4):899--916, 2006.

\bibitem{Knuth1997}
D.~E. Knuth.
\newblock {\em The art of computer programming}, volume~1.
\newblock Addison-Wesley, 3rd edition, 1997.

\bibitem{LevineSteinhardt1984}
D.~Levine and P.~J. Steinhardt.
\newblock Quasicrystals: a new class of ordered structures.
\newblock {\em Physical review letters}, 53(26):2477, 1984.

\bibitem{Lindenstrauss2001}
E.~Lindenstrauss.
\newblock Pointwise theorems for amenable groups.
\newblock {\em Inventiones mathematicae}, 146(2):259--295, 2001.

\bibitem{Littlewood1911}
J.~E. Littlewood.
\newblock The converse of {A}bel's theorem on power series.
\newblock {\em Proceedings of the London Mathematical Society}, 2(1):434--448,
  1911.

\bibitem{LSZVol1}
S.~Lord, F.~Sukochev, and D.~Zanin.
\newblock {\em Singular Traces, Volume 1: Theory}.
\newblock De Gruyter, 2021.

\bibitem{Lubotzky1994}
A.~Lubotzky.
\newblock {\em Discrete groups, expanding graphs and invariant measures},
  volume 125.
\newblock Springer Science \& Business Media, 1994.

\bibitem{NakamuraSakamoto2021}
Y.~Nakamura, R.~Sakamoto, T.~Mase, and J.~Nakagawa.
\newblock Coordination sequences of crystals are of quasi-polynomial type.
\newblock {\em Acta Crystallographica Section A: Foundations and Advances},
  77(2):138--148, 2021.

\bibitem{Nelson1986}
D.~R. Nelson.
\newblock Quasicrystals.
\newblock {\em Scientific American}, 255(2):42--51, 1986.

\bibitem{OKeeffe1995}
M.~O'Keeffe.
\newblock Coordination sequences for lattices.
\newblock {\em Zeitschrift f{\"u}r {K}ristallographie - Crystalline Materials},
  210(12):905--908, 1995.

\bibitem{PapaefstathiouRobaina2021}
I.~Papaefstathiou, D.~Robaina, J.~I. Cirac, and M.~C. Ba{\~n}uls.
\newblock Density of states of the lattice {S}chwinger model.
\newblock {\em Physical Review D}, 104(1):014514, 2021.

\bibitem{Pastur1980}
L.~A. Pastur.
\newblock Spectral properties of disordered systems in the one-body
  approximation.
\newblock {\em Communications in mathematical physics}, 75(2):179--196, 1980.

\bibitem{PasturFigotin1991}
L.~A. Pastur and A.~Figotin.
\newblock {\em Spectra of random and almost-periodic operators}, volume 297.
\newblock Springer, 1991.

\bibitem{PradodeOliveira2021}
R.~A. Prado, C.~R. de~Oliveira, and E.~C. de~Oliveira.
\newblock Density of states and {L}ifshitz tails for discrete 1{D} random
  {D}irac operators.
\newblock {\em Mathematical Physics, Analysis and Geometry}, 24(3):1--29, 2021.

\bibitem{SemenovSukochev2015}
E.~Semenov, F.~Sukochev, A.~Usachev, and D.~Zanin.
\newblock Banach limits and traces on $\mathcal{L}_{1,\infty}$.
\newblock {\em Advances in Mathematics}, 285:568--628, 2015.

\bibitem{ShechtmanBlech1984}
D.~Shechtman, I.~Blech, D.~Gratias, and J.~W. Cahn.
\newblock Metallic phase with long-range orientational order and no
  translational symmetry.
\newblock {\em Physical review letters}, 53(20):1951, 1984.

\bibitem{Shiryaev1996}
A.~N. Shiryaev.
\newblock {\em Probability}, volume~95.
\newblock Springer, 2nd edition, 1996.

\bibitem{ShutovMaleev2008}
A.~V. Shutov and A.~V. Maleev.
\newblock Quasiperiodic plane tilings based on stepped surfaces.
\newblock {\em Acta Crystallographica Section A: Foundations of
  Crystallography}, 64(3):376--382, 2008.

\bibitem{ShutovMaleev2015}
A.~V. Shutov and A.~V. Maleev.
\newblock Penrose tilings as model sets.
\newblock {\em Crystallography Reports}, 60(6):797--804, 2015.

\bibitem{ShutovMaleev2018}
A.~V. Shutov and A.~V. Maleev.
\newblock Coordination numbers of the vertex graph of a {P}enrose tiling.
\newblock {\em Acta Crystallographica Section A: Foundations of
  Crystallography}, 74(2):112--122, 2018.

\bibitem{ShutovMaleev2010}
A.~V. Shutov, A.~V. Maleev, and V.~G. Zhuravlev.
\newblock Complex quasiperiodic self-similar tilings: their parameterization,
  boundaries, complexity, growth and symmetry.
\newblock {\em Acta Crystallographica Section A: Foundations of
  Crystallography}, 66(3):427--437, 2010.

\bibitem{Simon1982}
B.~Simon.
\newblock Schr{\"o}dinger semigroups.
\newblock {\em Bulletin of the American Mathematical Society}, 7(3):447--526,
  1982.

\bibitem{Simon2005}
B.~Simon.
\newblock {\em Trace ideals and their applications}, volume 120 of {\em
  Mathematical Surveys and Monographs}.
\newblock American Mathematical Society, 2nd edition, 2005.

\bibitem{Stanley1986}
R.~P. Stanley.
\newblock {\em Enumerative Combinatorics Volume 1}.
\newblock Wadsworth \& Brooks/Cole, 1986.

\bibitem{Veselic2005}
I.~Veseli{\'c}.
\newblock Spectral analysis of percolation {H}amiltonians.
\newblock {\em Mathematische Annalen}, 331(4):841--865, 2005.

\bibitem{Wegner1981}
F.~Wegner.
\newblock Bounds on the density of states in disordered systems.
\newblock {\em Zeitschrift f{\"u}r {P}hysik {B} Condensed Matter}, 44(1):9--15,
  1981.

\bibitem{YakuboNakayama1987}
K.~Yakubo and T.~Nakayama.
\newblock Absence of the hump in the density of states of percolating clusters.
\newblock {\em Physical Review B}, 36(16):8933--8936, 1987.

\end{thebibliography}
\end{document}